\documentclass[11pt,reqno]{amsart}

\RequirePackage[OT1]{fontenc}
\usepackage{amsthm,amsmath,amssymb}
\numberwithin{equation}{section}
\theoremstyle{plain}
\newtheorem{theorem}{Theorem}[section]
\newtheorem{proposition}[theorem]{Proposition}
\newtheorem{lemma}[theorem]{Lemma}

\numberwithin{theorem}{section}

\newcommand{\upbar}[1]{\,\overline{\! #1}}

\newcommand{\nn}{\nonumber}
\DeclareMathOperator \sign{sign}
\newcommand{\mc}[1]{{\mathcal #1}}
\newcommand{\mf}[1]{{\mathfrak #1}}

\newcommand{\bb}[1]{{\mathbb #1}}

\renewcommand{\epsilon}{\varepsilon}

\newcommand{\id}{{1 \mskip -5mu {\rm I}}}

\newcommand{\bam}{\,\overline {\!m\!}\,}
\newcommand{\rme}{\mathrm{e}}
\newcommand{\rmi}{\mathrm{i}}
\newcommand{\rmd}{\mathrm{d}}
    
\begin{document}

\title[Front fluctuations for the Stochastic C-H equation]{Front fluctuations for the stochastic Cahn-Hilliard equation}

\author [L.\ Bertini] {Lorenzo Bertini}
\address{Lorenzo Bertini, Dipartimento di Matematica, SAPIENZA Universit\`a di Roma, P.le Aldo Moro 5, 00185 Roma, Italy} \email{bertini@mat.uniroma1.it}

\author [S.\ Brassesco] {Stella Brassesco}
\address{Stella Brassesco, Departamento de Matem\'aticas, Instituto Venezolano de Investigaciones Cient\'{\i}ficas, Apartado Postal 20632, Caracas 1020-A, Venezuela} \email{sbrasses@ivic.gob.ve}

\author [P.\ Butt\`a] {Paolo Butt\`a} 
\address{Paolo Butt\`a, Dipartimento di Matematica, SAPIENZA Universit\`a di Roma, P.le Aldo Moro 5, 00185 Roma, Italy} \email{butta@mat.uniroma1.it}



\begin{abstract}
We consider the Cahn-Hilliard equation in one space dimension, perturbed by the derivative of a space and time white noise of intensity $\epsilon^{\frac 12}$, and we investigate the effect of the noise, as  $\epsilon \to 0$, on the solutions when the initial condition is a front that separates the two stable phases. We prove that, given $\gamma< \frac 23$, with probability going to one as $\epsilon \to 0$, the solution remains close to a front for times of the order of $\epsilon^{-\gamma}$, and we study the fluctuations of the front in this time scaling. They are given by a one dimensional continuous process, self similar of order $\frac 14$ and non Markovian, related to a fractional Brownian motion and for which a couple of representations are given.  \end{abstract}


\maketitle

\section {Introduction}
\label{sec:1}
           
The kinematics of phase segregation for binary alloys can be described by the Cahn-Hilliard equation \cite{CH1,CH2,C},
\begin{equation}
\label{p1}
\partial_t u = -\Delta\Big(\frac 12 \Delta u - V'(u) \Big)\;,
\end{equation}
where $\Delta$ is the Laplacian and $V\colon \bb R \to \bb R$ is a symmetric double well potential which, for the sake of concreteness, is chosen as 
\begin{equation}
\label{V}
V(u) = \frac 14(u^2-1)^2\;.
\end{equation}
The scalar field $u=u(x,t)$ is an order parameter and represents the relative concentrations of the two species. The space homogeneous stationary solutions $u=\pm1$ are to be interpreted as  the pure phases of the alloy. In contrast with the Allen-Cahn flow \cite{AC}, the evolution governed by \eqref{p1} locally conserves the mass $\int\!\rmd x\, u(x,t)$. Indeed, \eqref{p1} can be viewed as the $H^{-1}$ gradient flow of the van der Waals free energy functional,
\begin{equation}
\label{AC}
\int\!\rmd x\, \Big[ \frac 14 |\nabla u|^2 + V(u) \Big]\;.
\end{equation}
In particular, the critical points of \eqref{AC} with mass conservation constraint are stationary solutions to \eqref{p1}. Moreover, when \eqref{p1} is considered in a bounded domain, its solutions converge, apart from exceptional initial conditions, to a minimizer of \eqref{AC} with the mass fixed by the initial datum and the yet unspecified boundary condition. By introducing an appropriate scaling parameter, the main issues from a heuristic point of view are the following ones. First, an early, relatively fast, stage of the dynamics (referred to as spinodal decomposition), in which the flow \eqref{p1} approaches to a critical point of \eqref{AC} by realizing a local separation of the pure phases $\pm 1$. A later, slow, stage of the evolution toward the minimizer, taking place in a small neighborhood of the unstable manifold of the critical point.

Let us focus on the one dimensional case. For a bounded domain, a detailed analysis of the slow evolution of patterns of the rescaled Cahn-Hilliard equation is given in \cite{ABF} and \cite{BX1,BX2}. More precisely, in a neighborhood of a stationary solution to \eqref{p1} having a given number of transition layers, the exponentially slow speed of the layer motion is determined. In addition, the existence of an unstable invariant manifold attracting solutions exponentially fast in the scaling parameter is established. 

Analogous issues can be posed when the equation \eqref{p1} is considered on the whole line without scaling parameter. By interpreting \eqref{AC} as an action functional, it is easy to show that \eqref{p1} admits ``droplet-shaped'' stationary solutions, i.e., profiles of the form $u_\mathrm{droplet}(x) = g(x-x_0)$ with $x_0\in \bb R$ and $g$ symmetric, monotone for $x>0$, and exponentially approaching its asymptotic value. The function $g$ can be determined up to an arbitrary parameter which plays the role of the mass constraint. The linear stability of such stationary solutions is determined by the spectrum of the linearization of \eqref{p1} around $u_\mathrm{droplet}$. By translation invariance, zero is an eigenvalue, which is responsible for the exponentially small motion in finite but large domains \cite{ABF,BX1,BX2}. Those results suggest that the remainder of the spectrum is bounded away from zero, but we are not aware if this has been proven.

Simpler stationary solutions are the ``kink-shaped'' profiles $\bam_{x_0}$, which describe a  transition between the pure phases $\pm 1$ at $\pm\infty$ with ``center'' $x_0\in \bb R$. By translation invariance $\bam_{x_0}(x) = \bam(x-x_0)$ with $\bam=\bam_0$. For the  specific choice \eqref{V} of the potential, $\bam(x) = \tanh(x)$. Its linear stability has been analyzed in  \cite{BKT} and \cite{H}. Again zero is a simple eigenvalue  but, in contrast both to the droplet and to the kink for the Allen-Cahn dynamics, it is an accumulation point of the spectrum. Hence, the manifold 
\begin{equation}
\label{mcm}
\mc M = \left\{\bam_{x_0},\; x_0\in\bb R\right\}
\end{equation}
is not exponentially attracting for the Cahn-Hilliard flow. Sharp estimates on the actual rate of convergence are proven in \cite{BKT} and \cite{H}: roughly speaking, an initial datum close to $\mc M$ relaxes to a front in $\mc M$ with a diffusive behavior. We remark that this is due to the fact that the domain is unbounded.

Both from a conceptual and a modelling point of view, the addition of a small random forcing term to \eqref{p1} appears natural. Clearly, the random force should preserve the local mass and, under suitable assumptions, can be taken to be Gaussian. We thus consider the stochastic partial differential equation,
\begin{equation}
\label{p2}
\partial_t u = -\Delta\Big(\frac 12 \Delta u - V'(u) \Big) + \sqrt {\epsilon}\, \nabla ( a_\epsilon \dot W)\;,
\end{equation}
where $\dot W$ is a space-time white noise, $a_\epsilon$ is a convenient space cut-off,  and $\epsilon>0$. In the framework of dynamical critical phenomena this is, with the choice \eqref{V}, a model for the evolution with conserved order parameter (Model B in the review \cite{HH}). In spite of the short-scale singularity of the forcing term and of the unbounded domain, in the one dimensional case this equation has a meaningful mild formulation for suitable $a_\epsilon$ (see the next section). We refer to the works \cite{AK,Cw,DD} for existence results on stochastic perturbations of Cahn-Hilliard equation in bounded domains. The effect of the noise on the motion of the transition layers analyzed in the aforementioned works \cite{ABF,BX1,BX2}, is studied in  \cite{ABK}. It is there shown that the random fluctuations dominate the exponentially slow deterministic dynamics and an effective system of stochastic ordinary differential equations for the motion of the layers is derived. 

The purpose of the present paper is to analyze the fluctuations of the kink profile $\bam$ due to the random noise in \eqref{p2}, in the limit $\epsilon \to 0$. Let us first review the corresponding results for the stochastic non-conservative Allen-Cahn equation $\partial _tu = \frac 12 \Delta u - V'(u) + \sqrt {\epsilon}\, \dot W$. If the initial datum is $\bam$, in \cite{BDP,BBDP,F}, it is shown that the solution at times $\epsilon^{-1}t$ stays close to $\bam_{\zeta_\epsilon(t)}$ for some random process $\zeta_\epsilon(t)$ which converges to a Brownian motion as $\epsilon\to 0$. To explain heuristically this result, let us regard the random forcing term as a source of independent small kicks, which we decompose along the directions parallel and orthogonal to $\mc M$. The orthogonal component is exponentially damped by the deterministic drift, while the parallel component, associated to the zero eigenvalue of the linearization around $\bam_\zeta$, is not contrasted and, by independence, sums up to a Brownian motion. 

For the Cahn-Hilliard dynamics this picture has to be completely modified, taking into account the following two related effects. The local mass conservation, which implies that fluctuations of the interface center can occur only in infinite volume, as the extra mass has ``to come from infinity'', and the absence of a spectral gap, which implies that perturbations in the direction orthogonal to $\mc M$ cannot be neglected. More precisely, the projection of the noise in the direction parallel to $\mc M$ vanishes, but its perpendicular component is not exponentially damped and gives rise, with a suitable delay, to the front fluctuations. On heuristic grounds, we expect that the increments of the resulting process are not independent but negatively correlated. Indeed, after a fluctuation in a given direction, the extra mass is reabsorbed, causing a successive fluctuation in the opposite direction. This mechanism is slower than that for the Allen-Cahn dynamics, and therefore a finite displacement of the interface should occur at times of the order $\epsilon^{-2}$. On the other hand, it is not at all clear whether the kink-like shape of the solution survives until so long times. If this is the case, as the limiting process is anyway the sum of approximately Gaussian increments, we expect it to be a self-similar of order $\frac 14$, non-Markovian  Gaussian process.

In this paper we show that, given $\gamma<\frac 23$,  with probability going to 1 as $\epsilon \to 0$, the solution to the equation \eqref{p2} with initial condition $\bam$
stays close to $\bam_{Z_\epsilon(t)}$ up to times of order  $\epsilon^{-\gamma}$, and we  describe the statistics of the kink fluctuations in this regime. More precisely, we prove that 
\[
u(\cdot,\epsilon^{-\gamma}t) \approx \bam_{Z_\epsilon(t)}\;, \qquad Z_\epsilon(t) \approx \epsilon^{\frac 12-\frac \gamma4} Z(t)\;,
\]
where the process $Z$ is the odd part of  a two sided  fractional Brownian motion with self-similarity parameter $\frac 14$, which is indeed a non-Markovian process, with negatively correlated increments.   As previously remarked, the above result suggests that the order of the time needed to obtain a finite displacement of the front should be $\epsilon^{-2}$, but the analysis requiered to reach this scaling is not clear.

Actually, the heuristic picture discussed before has been substantiated for the stochastic phase field equation \cite{BBBP}. In that case, due to the weak coupling between the order parameter and the phase field, there is a sharp separation between the instantaneous noise kicks and their delayed contribution to the front propagation; the analysis can thus be carried out up to the longer time scale. 

We notice that, in contrast to the kink fluctuations, the picture describing the fluctuations of the droplet for \eqref{p2} on the whole line should be analogous to that of the kink for the Allen-Cahn equation. In particular, the droplet fluctuations should become of order one at times $\epsilon^{-1}t$ and converge to a Brownian motion, due to the fact that the mass conservation is not present here, and the droplet can move freely. 

\section {Notation and results} 
\label{sec:2}

We consider, for each $\epsilon>0$, the process $u(x,t)$, $x\in \bb R$ and $t>0$, solution to the initial value problem for the  Cahn-Hilliard equation with a conservative stochastic perturbation,
\begin{equation}
\label{sche}
\begin{cases} \partial_t u= -\partial_x^2\Big(\frac 12 \partial_x^2 u - V'(u) \Big) + 
\sqrt {\epsilon}\, \nabla\big( a_{\epsilon} \dot W\big)\;, \\ u(x,0)=u_0(x)\;, \end{cases}      
\end{equation}
where $\partial_x^j=\frac{\partial^j}{\partial x^j}$ ($j\in\bb N$),  $V(u)$ as in \eqref{V}, and $a_{\epsilon}(x) =a(x\epsilon ^{\beta})$, for $\beta>2 $ and $a$ a $C^\infty$ positive function with $\mathrm{supp}(a) \subset [-1,1]$,  $a(0)=1$, and $\|a\|_{\infty}=1$. Finally, $\dot W=\dot W_{x,t}$ is a space-time white noise, defined on a probability space $\big( \Omega,\mf F,\bb P\big)$. The initial condition $u_0$ is taken  close to the centered interface $\bam(x)=\tanh(x)$.

The precise sense of the above equation is given by the integral equation obtained in terms of the Green function $G(x,y,t)$ corresponding to the operator $\rme^{-\frac 12 t\, \Delta^2}$,  where $\Delta^2 f=\partial _x^4 f$ when restricted to $f\in C^{\infty}$, 
\begin{equation}
\label{freeie}
\begin{split}
u(x,t) & = \int\! \rmd y\,  G(x,y,t)\,u_0(y) \\ & \quad + \int_0^t\! \rmd s \int\!\rmd y\, \partial_y^2 G(x,y,t-s) V'(u(y,s)) +  \sqrt{\epsilon}\, Y(x,t)\;.
\end{split}
\end{equation}
The last term above is the Gaussian process obtained formally as 
\[
Y(x,t)= \int_0^t \!\rmd s \int\!\rmd y\, G(x,y,t-s) \partial_y\big(a_{\epsilon}(y)\dot W_{y,s}\big)\;.
\]
Precisely, $Y(x,t)$ is the centred Gaussian process with covariance,
\begin{equation}
\label{cov}
\bb E\big(Y(x,t)Y(x',t')\big) = \int^{t\wedge t'}_0\!\rmd s \int\!\rmd y\,\partial _y G(x,y,t-s) \partial _y G(x',y,t'-s) a_{\epsilon}(y)^2\;.
\end{equation}

As usual, for $g\in L_2(\bb R \times [0,\infty))$, we denote $\dot W(g) = \int\! g(x,t)\, \rmd W_{x,t}$, omitting the variables $x,t$ in the integral when clear otherwise. The elements in the space $\Omega$ will be denoted by $\omega$. We also consider the filtration,
\[
\mf F_t = \sigma\{\dot W(A\times [0,t])\,;\; A\in \mf B(\bb R)\},\quad \mf B(\bb R) \mbox{ the Borel sets in $\bb R$.}
\]

In the sequel, we will denote by $C$ a generic positive constant, whose numerical value may change from line to line and from one side to the other in an inequality. The notation $a\land b$ $(a\lor b )$  stands for the minimum (maximum) between the real numbers $a$ and $b$. Given $p\in [1,\infty]$, we let $\|\cdot\|_p$ be the norm in $L_p(\bb R,\rmd x)$. We consider $C(\bb R_+)$ equipped with the (metrizable) topology of uniform convergence in compacts, and denote by $\Longrightarrow$   weak convergence of processes in that space.  Finally, to simplify the writing, we assume that the scaling parameter $\epsilon \in (0,1)$.

\smallskip
Our main results are stated as follows.

\begin{theorem}
\label{t1}
Given $0<\gamma<\frac 23$ and $T>0$, there exists a set $B(\epsilon,\gamma,T) \in \mf F$ such that $\bb P\big( B(\epsilon,\gamma,T)\big)\longrightarrow1$ as $\epsilon \to 0$, and, for $\omega\in B(\epsilon,\gamma,T)$, the stochastic Cahn-Hilliard equation \eqref{sche} with initial condition $\bam$ has a unique bounded  continuous solution $u(x,t)$, for $x\in \bb R$ and $t\le \epsilon^{-\gamma}T$. Moreover, there exists a one dimensional $\mf F _t$-adapted process $\zeta_\epsilon (t)$ such that
\begin{itemize}
 \item[\rm{(i)}] For each $\eta >0$,
\[
 \lim _{\epsilon \to 0} \bb P \Big(\sup_{t\le \epsilon ^{-\gamma}T} \left\| u(\cdot,t)-\bam_{\zeta_\epsilon(t))}\right\|_\infty > \epsilon^{(1-\gamma) \land \frac 12  - \eta} \Big)=0 \;.
 \]
\item[\rm{(ii)}] Consider the real process $X_\epsilon(t) := \epsilon^{-\frac 12 +\frac\gamma4 }\zeta_\epsilon(\epsilon^{-\gamma}t)$ . Then, 
\[
X_\epsilon(t)  \Longrightarrow (8\pi)^{\frac 14} r(t)\quad \mbox{as } \epsilon \to 0,
\] 
 where $r(t)$ is the one dimensional centered Gaussian process with covariance, 
\begin{equation}
\label{dcovr}
\bb E\big(r(t)\,r(t')\big) = \sqrt{t+t'}-\sqrt{t-t'}\;,  \qquad t\ge t'\;.
\end{equation}
\end{itemize}   
\end{theorem}

\begin{theorem}
\label{t2}
The one dimensional Gaussian  process $r(t)$ with covariance given by \eqref{dcovr} is a continuous process, self-similar of order $\frac 14 $. It admits  the following three representations, in terms of well known processes:
\begin{itemize}
\item  Let $\nu^{(H)}(t)$ be the usual two sided fractional Brownian motion with Hurst parameter $H$, and  $\nu_O^{(H)}$  its odd part. Then, 
\[  
r(t)= 2\,\nu_O^{(\frac 14)}(t)\;.
\]
 \item Let  $h(x,t)$ be  the solution to the heat equation in $\bb R$ with additive space-time white noise and initial condition $0$. Then,
 \[
r(t)= (2\pi)^{\frac 14} h(0,t)\;.
\]
\item Consider $b(s)$ a standard Brownian motion. Then, 
\[
r(t)=c\int_0^t\! \frac{s^{\frac14}}{(t^2-s^2)^{\frac14}}\, \rmd b(s)\;.
\] 
where the constant $c = \big(\frac 12 B(\frac 34, \frac 34)\big)^{-\frac12}$, with $B(\cdot,\cdot)$ the usual  Euler beta  function. 
\end{itemize}
\end{theorem}

The paper is organized as follows. In the next section, we discuss the properties of the Gaussian process $Y$, obtaining some \emph{sup}-norm estimates that allow to establish a local existence and uniqueness theorem for equation \eqref{freeie}. In Section \ref{sec:4}, we introduce a new integral equation, in terms of the kernel of the linearized equation around $\bam$. With the aid of estimates for this kernel derived in \cite{BKT}, we obtain estimates that are used in Section \ref{sec:5} to show stability of the front in a convenient time-scaling. In Section \ref{sec:6}, we study the Gaussian process $H$ appearing in the linearization about the front under proper time scaling, and prove Theorem \ref{t1}. Finally,  Theorem \ref{t2} is proved in Section \ref{sec:7}. Some technical proofs are reported in Appendix \ref{sec:A}. 
   
\section {Local existence and uniqueness of solutions}
\label{sec:3}

\subsection{The process $Y$}
\label{sec:3.1}

We investigate first the properties of the Gaussian process $Y$, which are deduced from that of its covariance \eqref{cov}. In the next lemma, we provide some properties of the Green function $G$ appearing in \eqref{freeie}, that will be useful later.
\begin{lemma}
\label{propfi}
The Green function $G(x,y,t)$ associated to $\rme^{-\frac 12 t \,\Delta^2}$ is given by 
\begin{equation}
\label{defphi}
G(x,y,t)=\frac{1}{t^{\frac 14}}\,\phi\Big(\frac{x-y}{t^{\frac 14}}\Big)\;, \quad \mbox{ where}\quad \phi(x)=\frac{1}{2\pi} \int\!\rmd \omega\, \rme^{-\frac 12 \omega^4} \rme^{\rmi \omega x}\;.
\end{equation}
Therefore, 
\begin{equation}
\label{dyg}
\partial_y G(x,y,t) = - \frac{1}{t^{\frac 12}}\, \phi'\Big(\frac{x-y}{t^{\frac 14}}\Big)\;, \;\;\, \partial_y^2 G(x,y,t) = \frac{1}{t^{\frac 34}}\,\phi''\Big(\frac{x-y}{t^{\frac 14}}\Big) \;.
\end{equation} 
Moreover, for each $\lambda>0$ and nonnegative integer $k$,
\begin{equation}
\label{norma}
\upbar c=\upbar c(\lambda,k) := \sup_{x\in\bb R} \Big|\rme^{\lambda|x|}\,\frac{\rmd^k\phi}{\rmd x^k}(x)\Big| < + \infty\;.
\end{equation}
Also, for any $H>0$ and $A\ge 1$, 
\begin{align}
\label{px}
\Big|\frac{\rmd^k\phi}{\rmd x^k}(x+H)-\frac{\rmd^k\phi}{\rmd x^k}(x)\Big| & \le  C H\Big( \rme^{-|x|} \id_{\{|x|>2H\}} + \frac{1}{1+H} \id_{\{|x|\le 2H\}}\Big)\;,
 \\
  \label{pt}  \Big|\frac{\rmd^k\phi}{\rmd x^k}(Ax)-\frac{\rmd^k\phi}{\rmd x^k}(x)\Big|
   & \le C \frac{A-1}{1+(A-1)|x|}\, |x|\rme^{-|x|} \le \frac{A-1}{A}\,\rme^{-|x|} \;.
\end{align}
For the second derivative, we have also that 
\begin{equation}
\label{2d1}
|\phi''(x+H)-\phi''(x)| \le C \frac{H^2+H|x|}{1+H^2+H|x|}\;.
\end{equation}
\end{lemma}

\begin{proof}  
The expression \eqref{defphi} follows by standard Fourier analysis. To prove \eqref{norma}, we observe that, as the entire function $f(z) =\exp\left(-\frac12 z^2+\rmi x z\right)$ satisfies $\max_{0\le\eta\le\lambda} |f(\pm R+\rmi \eta)| \to 0$ as $R\to \infty$, by  Cauchy's  theorem we have
\[
\phi(x) = \frac1{2\pi} \int\!\rmd \omega\, \rme^{-\frac 12 (\omega+\rmi\lambda) ^4} \rme^{\rmi (\omega+\rmi \lambda) x}\;,
\]
and hence
\[
\frac{\rmd^k\phi}{\rmd x^k}(x) = \frac{\rme^{-\lambda x}}{2\pi} \int\!\rmd \omega\, (\rmi \omega -\lambda)^k \rme^{-\frac 12 (\omega+\rmi\lambda) ^4} \rme^{\rmi\omega x}\;.
\]
As the integral on the right hand side is uniformly bounded in $x$ for each $\lambda$, the previous expression implies the bound \eqref{norma} for $x\ge 0$, and then also for $x<0$ as $\phi$ is an even function.

The estimates \eqref{px} and \eqref{pt} follow from \eqref{norma} with $\lambda=2$, as
\[
\begin{split}
\Big|\frac{\rmd^k\phi}{\rmd x^k}(x+H) -\frac{\rmd^k\phi}{\rmd x^k}(x)\Big| \le \upbar c \int_x^{x+H}\!\rmd y\, \rme^{-2 |y|} \le \begin{cases} \upbar c H\rme^{-|x|} & \mbox{if } |x|>2H\;, \\ \bar c (H\land 1) & \mbox{if } |x|\le 2H\;, \end{cases}
\end{split}
\]
and 
\[
\Big|\frac{\rmd^k\phi}{\rmd x^k}(Ax)-\frac{\rmd^k\phi}{\rmd x^k}(x)\Big|  \le \upbar c \, \Big|\int_x^{Ax}\!\rmd y\, \rme^{-2|y|}\Big| \le \upbar c\, \big((A-1)|x| \land 1\big)  \rme^{-|x|}
\]
(note that $a\land 1 \le 2a(1+a)^{-1}$ $\forall\, a\ge 0$). 

Finally, to prove \eqref{2d1}, we consider the integral expression for $\phi$ given in \eqref{defphi} and observe that 
\[
\left|\cos (\omega x) (\cos(\omega H)-1)\right| \le\frac{(\omega H)^2}{2}\land 2\;, \quad \left|\sin(\omega x)\sin(\omega H) \right| \le \omega^2 |x| H \land 1\;.
\]
Therefore,
\[
\begin{split}
& |\phi ''(x+H)-\phi''(x)| = \frac{1}{2\pi}\, \Big| \int\!\rmd \omega\, \rme^{-\frac 12 \omega^4} \omega^2 \rme^{\rmi\omega x} \left(\rme^{\rmi\omega H}-1\right)\Big|
 \\
& \le \frac{1}{2\pi} \int\!\rmd \omega\, \rme^{-\frac 12 \omega^4}\omega^2 \left(\left|\cos (\omega x) (\cos(\omega H)-1)\right|  + \left|\sin(\omega x)\sin(\omega H) \right| \right) 
\\
& \le C \int\!\rmd \omega\, \rme^{-\frac 12 \omega^4} \omega^2 \big[(H^2 \omega^2
+ \omega^2 |x| H) \land 1\big] \le C \big[ (H^2+H|x|) \land 1\big]\;,
\end{split}
\]
which proves \eqref{2d1} (using, as before, that $a\land 1 \le 2a(1+a)^{-1}$ $\forall\, a\ge 0$).
\end{proof}

The following proposition, whose proof is given in Appendix \ref{sec:A}, is a consequence of the estimates detailed in Lemma \ref{propfi}.

\begin{proposition}
\label{hoy}
Let $Y(x,t)$ be the Gaussian process with covariance \eqref{cov}. Then, for any $h>0$,\begin{align}
\label{vary}
&\bb E\,  Y(x,t)^2\le C t^{\frac 14}\;, \\ \label{hx} & \bb E \big( Y(x+h,t)- Y(x,t)\big)^2\le C h \log\big(1+h^{-1}t^{\frac 14}\big) \;, \\ \label{ht} &\bb E \big( Y(x,t+h)- Y(x,t)\big)^2\le C h^{\frac 14}\;.
\end{align}
\end{proposition}

Using well known results for Gaussian processes, \eqref{hx} and \eqref{ht} imply H\"older  continuity of the paths of  $Y(x,t)$ in both variables. From \eqref{vary}, it follows at once that $Y(x,t)\in L^2(\bb R\times [0,T], \rmd\mu)$ with probability one, for any $T>0$ given, and $\mu$ a finite measure  on $\bb R\times [0,T]$. The next results provide more precise global information on the paths, that will be useful in establishing uniqueness and existence of solutions to \eqref{freeie}. Recall that, although not explicit  in the notation, the process $Y$ depends on $\epsilon$.

Let us define, for each positive $\epsilon$, $\gamma$, and $T$, the set 
\begin{equation}
\label{defmct}
\mc T_\epsilon =\{(x,t)\colon x\in \bb R,\; t\in [0,T\epsilon ^{-\gamma}]\}\;.
\end{equation}

\begin{proposition}
\label{supy}
Consider the set $\mc T_\epsilon$ as defined in \eqref{defmct}. Then the process $Y(x,t)$ satisfies the following properties.   
\begin{itemize}
\item[\rm{(i)}]
For each $\xi>0$ there exists a constant $C>0$ such that, for any $\epsilon\in(0,1)$, 
\begin{equation}
\label{esup}
\bb E \sup_{(x,t)\in \mc T_\epsilon} | Y(x,t)| <C \epsilon ^{-\frac{\gamma}{8} -\xi}\;. 
\end{equation}
\item[\rm{(ii)}]
 Given $0<\gamma <4$,
\begin{equation}
\label{i} 
\bb P\big(\sup_{(x,t)\in \mc T_\epsilon}\sqrt \epsilon \,| Y(x,t)|<\infty\big)=1\;.
\end{equation}
\item[\rm{(iii)}]
Given $0<\gamma <4$, for any $0<\xi<\frac{4-\gamma}{16}$ there exist $\zeta>0$ and a constant $C>0$ such that,   
\begin{equation}
\label{ii}
\bb P \big(\sup_{(x,t)\in \mc T_\epsilon}\sqrt \epsilon \,| Y(x,t)|>\epsilon ^{\xi}\big) \le C \rme^{-\epsilon^{-\zeta}}\;.
\end{equation}
\end{itemize}
\end{proposition}

\begin{proof} 
We first show that items (ii) and (iii) follow  from item (i). Indeed, by \eqref{esup} we have, 
\begin{equation}
\label{esupe}
\bb E \big(\sup_{(x,t)\in \mc T_\epsilon}\sqrt \epsilon \, | Y(x,t)|\big)\,<\,\epsilon ^{\frac{4-\gamma}{16}},
\end{equation}
which implies \eqref{i}. Next, from \eqref{vary},
\begin{equation}
\label{supvar}
\sigma^2 _{\mc T_\epsilon}:=\sup_{(x,t)\in \mc T_\epsilon}{\rm Var} \big(Y(x,t)\big)\le C\epsilon^{-\frac \gamma 4}\;, 
\end{equation}
so that, from \eqref{esupe}, \eqref{supvar}, and Borell's inequality \cite{adler},  
\begin{equation}
\begin{split}
\label{bineq}
\bb P\big(\sup_{(x,t)\in \mc T_\epsilon} \sqrt \epsilon | Y(x,t)|>\epsilon^{\xi}\big) & \le 4  \exp \Big(-\frac{\big(\epsilon^{-\frac 1 2+\xi} -\epsilon^{-\frac 12+\frac{4-\gamma} {16}}\big)^2}{2 \sigma^2_{\mc T_\epsilon}}\Big) \\  & \le 4  \exp \big(-C\,\epsilon^{-1+2\xi+\frac \gamma 4}\big)\;, 
\end{split}
\end{equation}
which implies \eqref{ii}. 

We are left with the proof of item (i). Without loss of generality we can suppose $T=1$ and $\epsilon\in \big(0,\frac 12\big)$, what we do. Consider the pseudo-metric $d$ on $\mc T_\epsilon$  defined by 
\[
d\big((x,t),(y,s)\big)=\big(\bb E \big(Y(x,t)-Y(y,s)\big)^2\big)^{\frac 12}\;,
\]
let $B_\delta (x,t)$ denote the $d$-ball of radius $\delta$ and center $(x,t)$, and  $N(\delta)$ the minimum number of $d$-balls of radius $\delta$ needed to cover $\mc T_\epsilon$. From entropy estimates \cite{adler} we know that there is a universal positive  constant $\kappa$ such that
\begin{equation}
\label{ent}
\bb E \sup_{(x,t)\in \mc T_\epsilon} | Y(x,t)| < \kappa \int_0^{\mbox{diam } \mc T_\epsilon}\!\rmd\delta\, \big(\log N(\delta)\big)^{\frac 12}\;.
\end{equation}
Recall that $Y(x,0)\equiv 0$. It follows from \eqref{vary} that the diameter of $\mc T_\epsilon$ satisfies
$\mbox{diam } \mc T_\epsilon \le C_0 \epsilon^{-\frac{\gamma}{8}}$. To estimate $N(\delta)$ we let $R$ be a positive parameter, to be fixed later as a function of $\delta$ and $\epsilon$. Proceeding as in \eqref{ey2}, with the aid of \eqref{norma} with $\lambda=1$ and after recalling the definition of $a_{\epsilon}$, if $|x|>\epsilon^{-\beta}+R$ and $t\le \epsilon^{-\gamma}$, we have,
\[
\begin{split}
\bb E & Y(x,t)^2 \le \int_0^t\!\rmd s\, s^{-\frac 34}\int_{s^{-\frac 14}(x-\epsilon^{-\beta})}^{s^{-\frac 14}(x+\epsilon^{-\beta})}\!\rmd y\,
\phi'(y)^2 \le \frac{\upbar c}2 \int_0^{\epsilon^{-\gamma}}\!\rmd s\, s^{-\frac 34} \, \rme^{-2\,s^{-\frac 14}R} \\ & \le \frac{\upbar cR}2 \rme^{-\epsilon^{\frac \gamma 4}R} \int_0^{R^{-4}\epsilon^{-\gamma}}\!\rmd\tau\, \tau^{-\frac 34}\, \rme^{-\tau^{-\frac 14}}
 \le \frac{\upbar cR}2 \rme^{-\epsilon^{\frac \gamma 4}R} C \big(R^{-1}\epsilon^{-\frac \gamma 4}\big) 
 \\
  & = C_1 \rme^{-\epsilon^{\frac \gamma 4}R}\,\epsilon^{-\frac \gamma 4}\;,
\end{split}
\]    
with $C_1>0$.  

 Let $R$ be the unique non negative  solution to 
 $ C_1 \rme^{-\epsilon^{\frac \gamma 4}R}\,\epsilon^{-\frac \gamma 4}=\delta^2$, which clearly  exists  if $\delta^2 \le C_1 \,\epsilon^{-\frac \gamma 4}$, and it is 
\begin{equation}
\label{R}
R=\epsilon^{-\frac \gamma 4}\Big|\log \frac{\delta^2 \,\epsilon^{\frac \gamma 4}}{C_1}\Big|
\end{equation}
In particular, for each $\delta$ such that $\delta^2 < C_1 \,\epsilon^{-\frac \gamma 4}$ and 
$R$ as in \eqref{R},  the set 
$\mc R_1=\mc T_\epsilon\cap \{x\colon |x|>\epsilon ^{-\beta}+R\}$ is contained in $B_\delta(0,0)$, so it is covered with one ball of radius $\delta$. If
 $\delta^2 > C_1 \,\epsilon^{-\frac \gamma 4}$, the corresponding set with $R=0$ is covered by 
 one ball of radius $\delta$.   Let us consider the rest of the parameter set, $\mc R_2=\mc T_\epsilon\cap \{x\colon |x|\le\epsilon ^{-\beta}+R\}$. Let $(x_0,t_0)\in \mc R_2$, and denote by $Q(x_0,t_0)=\{(x,t) :|x-x_0|\le b,|t-t_0|\le b^4\}$ the rectangle of sides $b$ and $b^4$ in the usual metric, for $b$ that will be conveniently chosen. If $(x,t)\in Q(x_0,t_0)\cap \mc T_\epsilon$, from \eqref{hx} and \eqref{ht} we get that, 
\[
\begin{split}
\bb E \big(Y(x,t)  - & Y(x_0,t_0)\big)^2 \le C\big(|x-x_0|\log(1+ |x-x_0|^{-1}t^{\frac 14})
 + |t-t_0|^{\frac 14}\big) \\ & \le  C \big(|x-x_0|\log(1+ |x-x_0|^{-1}) + |x-x_0| \log \epsilon^{-1} + b\big) 
\\ & \le C_2 \log\epsilon^{-1} (\sqrt b\,\id_{\{b\le 1\}} + b\,\id_{\{b> 1\}}\big)\;,
\end{split}
\]
with $C_2>0$ (recall we are assuming $\epsilon^{-1} \ge  2$). Then, choosing
\begin{equation}
\label{b}
b  = \Big(\frac{\delta^2}{C_2\log\epsilon^{-1}}\Big)^2 \id_{\{\delta^2\le C_2\log\epsilon^{-1}\}} + \frac{\delta^2}{C_2\log\epsilon^{-1}} \id_{\{\delta^2> C_2\log\epsilon^{-1}\}}\;,
\end{equation}
we obtain $Q(x_0,t_0)\subset B_\delta(x_0,t_0)$. It is now clear that $\mc R_2$ is covered by $b^{-5} \epsilon^{-\gamma}(\epsilon^{-\beta}+R)$ rectangles (convenient translations of $Q(x_0,t_0)$). We conclude that $N(\delta)\le 1+b^{-5}\epsilon^{-\gamma}(\epsilon^{-\beta}+R)$. Then, noticing that \eqref{R} implies 
$R\le C  \epsilon^{-\frac \gamma 4}\,(|\log\epsilon|+ |\log\delta|)$, by  \eqref{b} we get,
\[
\begin{split}
& N(\delta) \le 1 + C\epsilon^{-\gamma} \big(\epsilon^{-\beta} 
+\epsilon^{-\frac\gamma 4}( |\log\epsilon|+ |\log\delta|)\big) \\ & \qquad \times \Big(\frac{\big(C_2\log\epsilon^{-1}\big)^{10}}{\delta^{20}}\id_{\{\delta^2\le C_2\log\epsilon^{-1}\}} + \frac{\big(C_2\log\epsilon^{-1}\big)^5}{\delta^{10} }\id_{\{\delta^2> C_2\log\epsilon^{-1}\}}\Big)\;.
\end{split}
\]
For $\delta \le \mbox{diam } \mc T_\epsilon \le C \epsilon^{-\frac{\gamma}{8}}$, the last estimate gives 
\[
N(\delta)\le 1 + C\, \epsilon^{-A}\delta^{-21}\;,
\] 
with $A>0$ sufficiently large. Substitution of this in  \eqref{ent} yields
\begin{equation}
\label{enti} 
\begin{split}
& \int_0^{\mbox{diam } \mc T_\epsilon}\!\rmd\delta\, \big(\log N(\delta)\big)^{\frac 12} \le  \int_0^{C_0\epsilon^{-\frac \gamma 8}}\!\rmd\delta\, \big(\log(1 + C\epsilon^{-A}\delta^{-21}) \big)^{\frac 12} \\ & \qquad \qquad \qquad = \frac{C^{\frac{1}{21}}\epsilon^{-\frac{A}{21}}}{21}  \int_{CC_0^{-21}\epsilon^{\frac{21\gamma}8-A}}^\infty\!\rmd u\, \big(\log (1+u)\big)^{\frac 12} u^{-\frac{22}{21}} \\ & \qquad \qquad \qquad \le C \epsilon^{-\frac\gamma 8 -A\rho}\;,
\end{split}
\end{equation}
where in the last inequality we estimated the integrand  by $u^{-\frac{22}{21}+\rho}$ with $\rho\in(0,\frac{1}{21})$. The estimate \eqref{esup} now follows from \eqref{ent} and \eqref{enti} by choosing $\rho<A^{-1}\xi$ for the given $\xi$. 
\end{proof} 

\subsection{Existence and uniqueness of solutions.}
\label{sec:3.2}

From the previous results, it is not difficult to prove the existence of a unique  continuous solution to the integral equation \eqref{freeie} for $t\le T_0$ if $T_0$ is small enough. To that end, denote by $g$ and $\mc G$ the following operators, defined in terms of the Green function $G$,
\begin{equation}
\begin{split}
\label{defcalg}
g u_0(x,t) & = \int\!\rmd y\, G(x,y,t) u_0(y)\;,  \\ \mc G F(x,t) & = \int_0^t\!\rmd s \int\!\rmd y\, \partial_y^2 G(x,y,t-s) F(y,s)\;.
\end{split}
\end{equation}
so that equation \eqref{freeie} reads,
\begin{equation}
\label{suceq}
u=gu_0 + \mc G \big(V'(u)\big)+\sqrt\epsilon\, Y\;.
\end{equation}
 
\begin{proposition}
\label{locuye} Given $u_0$  continuous, $\|u_0\|_\infty =M<\infty$, there exists a time $T_0$ (depending on $\|Y\|_\infty$ and on $M$)  such that the equation \eqref{suceq} has a unique continuous bounded  solution on $\bb R\times [0,T_0]$. 
\end{proposition}

\begin{proof}
Denote $q=g u_0+\sqrt\epsilon\,Y$ and consider, for each $T>0$ fixed, the set 
\[
\mc C=\{v\in C(\bb R\times [0,T])\colon \|v\|_{\infty,T}<2\|q\|_{\infty,T}\}\;,
\]
where, for any $t>0$, $\|v\|_{\infty,t} := \sup\{|v(x,s)|\colon (x,s)\in\bb R\times [0,t]\}$. 
Consider on $\mc C$ the function $v\mapsto F(v) = \mc G \big(V'(v)\big)+q$, and observe that if $v\in \mc C$ then $\|V'(v)\|_{\infty,T} \le 8 |q\|_{\infty,T}(|q\|_{\infty,T}^2+1)$. Therefore, from \eqref{dyg} and \eqref{norma}, for any $t\in [0,T]$,
\begin{equation}
\begin{split}
\label{lineq}
|\mc G \big(V'(v)\big)(x,t)| & \le \int_0^t\!\rmd s \int\!\rmd y\, \big|\partial_y^2G(x,y,t-s) V'(u(y,s))\big| \\ & \le C \|q\|_{\infty,T}(\|q\|_{\infty,T}^2+1) t^{\frac 12},
\end{split}
\end{equation}
Then, choosing $T=T_0$ small enough, from \eqref{lineq} we see that $F(\mc C)\subseteq \mc C$.
 Moreover, the Picard iterates given by $v_0=q$, $v_{n+1}=F(v_n)$, form a Cauchy sequence on $\mc C$ (with sup norm) if $T_0$ is small, that converges to a limit $u$, which is a solution to \eqref{suceq}. To prove uniqueness, fix a realization of $Y$ and suppose that $u,\tilde u$ are continuous bounded solutions on $\bb R\times [0,T]$ with the same initial condition $u_0$ and same realization of $Y$. By \eqref{suceq}, \eqref{norma}, and H\"older inequality, for any $t\in [0,T]$,
\[
\begin{split}
\|u-\tilde u\|_{\infty,t} & \le C\int_0^t\!\rmd s\, \frac{1}{\sqrt{t-s}} \|u-\tilde u\|_{\infty,s}
\\
 & \le C\, \Big(\int_0^t\!\rmd s\, (t-s)^{-\frac 23}\Big)^{\frac 34}\Big(\int_0^t\!\rmd s\, \|u-\tilde u\|_{\infty,s}^4 \Big)^{\frac 14}\;,
\end{split}
\]
and therefore,
\[
\|u-\tilde u\|_{\infty,t}^4 \le C T \int_0^t\!\rmd s\, \|u-\tilde u\|_{\infty,s}^4\;,
\]
which implies that $u=\tilde u$ on $\bb R\times [0,T]$. 
\end{proof}

\section {Another integral equation}
\label{sec:4}

We introduce next a different  integral equation equivalent to \eqref{sche}, which is more convenient to analize the stability of $\bam$. Following \cite{BKT}, we consider the kernel arising from a convenient linearization of equation \eqref{sche} around $\bam$. Precisely, let
\begin{equation}
\label{ic}
u_0=\bam+h(x)\quad \mbox{with }\quad h=\partial _x f \mbox{ for some $f$ satisfying  } f(\pm \infty)=0,
\end{equation}
and denote by 
\[
Lu=\frac 12 \partial_x^2\,u -V''(\bam)u\;,
\]
the linearization around $\bam$ of the non-linear operator $u\mapsto \frac 12 \partial_x^2 u-V'(u)$. If $u$ solves \eqref{sche} for $u_0$ satisfying \eqref{ic}, then $v=u-\bam$ satisfies,
\[
\begin{cases} \partial _tv = -\partial_x^2 Lv + \partial_x^2 \big(3\bam v^2+v^3\big) + \sqrt \epsilon\, \nabla\big(a_{\epsilon} \dot W\big)\;,    
\\ v(x,0)=h(x)\;.
\end{cases}    
\]
Again, the previous equation is to be understood as the following integral equation for $v$ in terms of the Green function $K(x,y,t)$ corresponding to the operator $\rme^{-t\,\partial_x\,L\partial_x}$,
\begin{equation}
\label{linie}
\begin{split}
v(x,t) & = \int\!\rmd y\, \partial_x K(x,y,t) f(y) \\ & \quad -\int_0^t\!\rmd s \int\!\rmd y\, \partial_x\partial_y K(x,y,t-s)\big(3\bam v^2+v^3\big) (y,s) \\ & \quad + \sqrt \epsilon \int_0^t\int \partial_x K(x,y,t-s) a_\epsilon(y)\,  \rmd W_{y,s}\;.
\end{split}
\end{equation}
Let us now consider the right hand side above. The following expression for $K(x,y,t)$ can be deduced from the analogous one in \cite[Propositions 3.1 and 3.2]{BKT}, where the case $L = \partial_x^2 -\frac 12 V''(\bam)$ is considered.

\begin{proposition}
\label{BKT}
There exists $t_0>0$ such that the kernel $K(x,y,t)$ satisfies,
\begin{equation}
\label{defK}
K(x,y,t) = \begin{cases} K_\infty(x-y,t)+\tilde K(x,y,t)& \quad \mbox{ if } t\in (0,t_0)\;, \\ K^\ast (x,y,t)+k(x,y,t) & \quad\mbox{ if } t\ge t_0\;, \end{cases}
\end{equation}
where 
$K_\infty(x-y,t)$ is the kernel associated with $\rme^{-t\,\partial_x(\frac 12 \partial_x^2-2)\,\partial_x}$ and  the following estimates hold. For $i,j\in \{0,1\}$, there exists a constant $C>0$ such that  
\begin{align}
\label{Ki}
|\partial_x^i\partial_y^j K_\infty(x-y,t)|&\le C t^{-\frac{1+i+j}{4}}\, \exp(-2^{\frac 14}\, t^{-\frac 14}|x-y|) \;, 
\\ \label {Kt} |\partial_x^i\partial_y^j \tilde K(x,y,t)|&\le C t^{-\frac{i+j}{4}}\,
 \exp(-2^{\frac 14}\, t^{-\frac 14}|x-y|)\;,
\end{align}
and
\begin{equation}
\label{kstar}
\begin{split}
K^\ast (x,y,t) & = \frac{1}{\sqrt{2\pi t}} \Big\{ - \frac 12\,\varphi(x)\varphi(y) \sign (xy) + \frac 12 \rme^{-\frac{y^2}{8t}} \varphi(x)\sign (xy) \\ & \qquad + \frac 12 \rme^{-\frac{x^2}{8t}}\varphi(y) \sign (xy) - \rme^{-\frac{(x+y)^2}{8t}} \id_{\{\sign (xy)=1\}}\Big\}\;, 
\end{split}
\end{equation}
where 
\begin{equation}
\varphi(x) = \begin{cases} 1-\bam(x) \quad \mbox { if } x\ge 0\;, \\ 1+\bam(x) \quad \mbox { if } x\le 0\;.\end{cases}
\end{equation}
while, concerning the kernel $k(x,y,t)$, for $i,j\in \{0,1\}$ there exist $\mu>0$ and $C>0$ such that, 
\begin{equation}
\label{ddt}
|\partial_x^i\partial_y^j k(x,y,t)|\le \frac{C}{t} \exp\big(-\mu t^{-\frac 12} |x-y|\big) \;.
\end{equation}
\end{proposition}

For future reference, we compute $\partial_x K^\ast (x,y,t)$ and $\partial_y\partial_x K^\ast (x,y,t)$,
\begin{equation}
\label{ks}
\begin{split}
& \partial_x K^\ast (x,y,t) = \frac{1}{\sqrt{2\pi t}} \Big\{\frac 12 \bam'(x)\varphi(y)\sign(y)-\frac 12 \bam'(x)\rme^{-\frac{y^2}{8t}}\sign(y) \\ & \qquad  - \frac{x}{8t} \rme^{-\frac{x^2}{8t}}\varphi(y)\sign(x) \sign(y) + \rme^{-\frac{(x+y)^2}{8t}}\frac{(x+y)}{4t} \id_{\{\sign (xy)=1\}}\Big\}\;,
\end{split}
\end{equation}
\begin{equation}
\label{dyxk}
\begin{split}
& \partial_x \partial_y K^\ast(x,y,t) = \frac{1}{\sqrt{2\pi t}}\Big\{-\frac 12 \bam'(x)\bam'(y)+\bam'(x) \frac {y}{8t} \rme^{-\frac{y^2}{8t}} \sign(y) \\ & \quad \qquad + \frac {x}{8t} \rme^{-\frac{x^2}{8t}}  \bam'(y)\sign(x)+ \frac{1}{4t} \rme^{-\frac{(x+y)^2}{8t}}\id _{\{\sign (xy)=1\}} \\ & \quad \qquad - \frac{(x+y)^2}{16t^2} \rme^{-\frac{(x+y)^2}{8t}}\id _{\{\sign (xy)=1\}}\Big\}\;. 
\end{split}
\end{equation}
Let us denote by $H$ the Gaussian process on the last line of \eqref{linie},
\begin{equation}
\label{defH}
H(x,t)=\int_0^t\int \partial_x K(x,y,t-s)\,a_{\epsilon}(y)\, \rmd W_{y,s}\;.
\end{equation}
More precisely, $H(x,t)$ is the Gaussian process with covariance,
\begin{equation}
\label{covH}
\begin{split}
& \bb E\big(H(x,t)H(x',t')\big) \\ & \qquad \quad =\int^{t\wedge t'}_0\!\rmd s \int\!\rmd y\, \partial _x K(x,y,t-s) \partial _{x'} K(x',y,t'-s) a_{\epsilon}(y)^2 \;.
\end{split}
\end{equation}  
An expression for $H$ in terms of $Y$ also holds. Indeed, from \eqref{sche},
\begin{equation}
\label{aeq}
\partial _t(H-Y) = - \partial_x^2 \Big(\frac 12 \partial_x^2 (H-Y)-\partial_x^2(V''(\bam) H \Big)\;.
\end{equation} 
Recalling \eqref{defcalg}, and solving in terms of the Green functions $G$ and $K$ respectively we obtain,
\begin{equation}
\label{HK}
\begin{split}
H & = Y + G \partial_x^2 \big(V''(\bam)H) = Y + \mc G \big(V''(\bam) H)\;,
\\ H & = Y + \nabla K \nabla\big(V''(\bam) Y\big)\;. 
\end{split}
\end{equation}
It is not difficult to see that the process  $H$ is bounded and continuous as long as $Y$ is bounded and continuous. Moreover, $u$ is a solution to \eqref{freeie} with initial condition as in \eqref{ic} if and only if $v=u-\bam$ is a solution to \eqref{linie}. The next proposition, whose proof is given in  Appendix \ref{sec:A},  gives estimates on the increments of $H$. 

\begin{proposition}
\label{pcov}
Let $H(x,t)$ be as in \eqref{defH}. Then, for any $h>0$, 
\begin{align}
\label{varH}
&\bb E H(x,t)^2 \le C \Big(t^{\frac 14}\id _{\{t\le 1\}}+\big(1+\bam'(x)^2\, t^{\frac 12}\big) \id_{\{t>1\}}\Big)\;, \\ \label{hxH} & \bb E \big( H(x+h,t)- H(x,t)\big)^2\le C h\big(ht^{\frac 12} \big(1+ h^2+ |x|^2\big) + \log\big(1+h^{-1}t^{\frac 14}\big)\big)\;, \\ \label{htH} &\bb E \big( H(x,t+h)- H(x,t)\big)^2
\le C \big(h^{\frac 14}+h^{\frac 32} + (h+h^{\frac12}) t^{\frac 12}\big)\;.
\end{align}
\end{proposition}

The  estimate \eqref{varH} for the variance of the process $H(x,t)$ is uniform in $x$. 
However, we will need more precise estimates for $x$ large, which are considered in the next lemma. Recall that, as was already observed for the process $Y$, the  process $H$ depends on $\epsilon$ through $a_\epsilon$, see \eqref{defH}, although the dependence is  not explicit in the notation.
   
\begin{lemma}
\label{lxg}
For each $\delta>0$, $R>\epsilon^{-\frac{11\gamma}{10}} \delta^{-2}+1$, and any $\epsilon$ sufficiently small, the process $H$ satisfies,
\begin{equation}
\label{ex}
\sup_{|x|\ge R +\epsilon^{-\beta},\, t\le \epsilon^{-\gamma}T} \bb E H(x,t)^2 \le \delta ^2\;.
\end{equation}
\end{lemma}
 
\begin{proof}
Recalling \eqref{defH} and \eqref{defK}, 
\begin{equation}
\label{EH}
\begin{split}
\bb E H(x,t)^2 & = \int^{t\land t_0}_0\!\rmd s \int\!\rmd y\, \big(\partial_x K_\infty(x,y,s) + \partial_x \tilde K(x,y,s)\big)^2 a_\epsilon(y)^2 \\ & \quad + \int_{t\land t_0}^t\!\rmd s \int\!\rmd y\, \big(\partial_x K^{\ast}(x,y,s)+\partial_x k(x,y,s)\big)^2 a_\epsilon(y)^2\;.
\end{split}
\end{equation}
Therefore, from \eqref{Ki}, for $R>1$ and $|x|>R+\epsilon ^{-\beta}$, 
\begin{equation}
\label{e1}
\begin{split}
& \int^{t\land t_0}_0\!\rmd s \int\!\rmd y\, \big(\partial_x K_\infty(x,y,s)\big)^2   a_\epsilon(y)^2 
\le C \int^{t\land t_0}_0\!\rmd s\, \frac 1s \int^{\epsilon^{-\beta}}_{-\epsilon^{-\beta}}\!\rmd y\, \rme^{-s^{-\frac14}|x-y|} \\& \qquad \qquad \qquad \le C \int^{t\land t_0}_0\!\rmd s\,  \frac{1}{s^{\frac 34}}\rme^{-R s^{-\frac 14}} \le C \frac{(t\land t_0)^{\frac 12}}{R}\;.
\end{split}
\end{equation}
Similar computations yield, from \eqref{Kt} and \eqref{ddt} respectively, 
\begin{align}
\label{e2}
& \int^{t\land t_0}_0\!\rmd s \int\!\rmd y\, \big(\partial_x {\tilde K}(x,y,s)\big)^2  a_\epsilon(y)^2 \le C \frac{t\land t_0}{R}\;,  \\ \label{e3} &\int_{t\land t_0}^t\!\rmd s \int\!\rmd y\, \big(\partial_x k(x,y,s)\big)^2 a_\epsilon(y)^2 \le \frac{C}{R}\;.
\end{align}
With the aid of \eqref{ks} we estimate the term involving $\partial_x K^{\ast}$ in \eqref{EH}. Observe that $\bam '(x)\le \rme^{-|x|}$, and then 
\begin{equation}
\label{e4}
\begin{split}
&\int_{t\land t_0}^t \!\rmd s\, \frac{\bam'(x)^2}{8\pi s} \int\!\rmd y\,\big(\varphi(y)\sign(y)-\rme^{-\frac{y^2}{8s}}\sign(y)\big)^2 a_\epsilon(y)^2\\ & \qquad \qquad \qquad \le C \rme^{-2|x|}\big(\log t-\log (t\land t_0) + t^{\frac 12} \big).
\end{split}
\end{equation}
 Finally, 
\begin{align}
\label{e5}
&\int_{t\land t_0}^t \!\rmd s\, \int\!\rmd y\, \Big(\frac{ \varphi(y)}{\sqrt{2\pi s}}\frac{x}{4s}\,\rme^{-\frac{x^2}{8s}}\Big)^2 a_\epsilon(y)^2 \le C \int_{t\land t_0}^t\!\rmd s\, \frac{x^2}{s^3}\rme^{-\frac{x^2}{4s}} \le C \frac{1}{x^2}\;, \\ \label{e6} & \int_{t\land t_0}^t \!\rmd s\, \int\!\rmd y\, \Big(\frac{(x+y)}{4s\,\sqrt{2\pi s}} \rme^{-\frac{(x+y)^2}{8s}}\Big)^2 \id_{\{\sign (xy)=1\}} a_\epsilon(y)^2 \le C \frac {1}{|x|}\;.
\end{align} 
From \eqref{e1}-\eqref{e6}, \eqref{ex} follows for $R>\epsilon^{-\frac{11\gamma}{10}}\delta^{-2}+1$.  
\end{proof}

The continuity of the process $H(x,t)$ in both variables  follows from Proposition \ref{pcov}. We can also obtain estimates for the supremum of $H$.

\begin{proposition}
\label{suph}
Consider the set $\mc T_\epsilon$ as defined in \eqref{defmct} with $T>0$ and $\gamma>0$. Then the process $H(x,t)$ satisfies the following properties.   
\begin{itemize}
\item[\rm{(i)}]
For each $\xi>0$ there exists a constant $C>0$ such that, for any $\epsilon\in(0,1)$, 
\begin{equation}
\label{esupH}
\bb E \sup_{(x,t)\in \mc T_\epsilon} | H(x,t)| <C \epsilon ^{-\frac{\gamma}{4} -\xi}\;.
\end{equation}
\item [\rm{(ii)}]
\begin{equation}
\label{ih}
\bb P\big(\sup_{(x,t)\in \mc T_\epsilon}\epsilon^{\frac \gamma 4+\xi}| H(x,t)|<\infty\big)=1\;.
\end{equation}
\item[\rm{(iii)}]
There exist $\zeta>0$ and a constant $C>0$ such that,
\begin{equation}
\label{iih}
\bb P \big(\sup_{(x,t)\in \mc T_\epsilon} |H(x,t)|>\epsilon^{-\frac \gamma 4-\xi}\big) \le C \rme^{-\epsilon^{-\zeta}}\;.
\end{equation}
\end{itemize}
\end{proposition}

\begin{proof} We omit the details of the proof: item (i) follows from Proposition \ref{pcov} and Lemma \ref{lxg}, adapting the proof of Proposition \ref{supy}. Items (ii) and (iii) follow from  \eqref{varH} and \eqref{esupH}, proceeding again as in the demonstration of Proposition \ref{supy}. 
\end{proof}
 
\section{Stability of $\bam$}
\label{sec:5}
In this section we prove the stability of the front $\bam$ up to times of the order $\epsilon^{-\gamma}$, with $\gamma<\frac 23$. A precise statement is given in Proposition \ref{gapprox}. 

\begin{lemma}
\label{sup1ddk}
There exists a constant $M>0$ such that  for any $T\ge 0$, 
\begin{equation}
\label{cddk}
\sup_{x\in \bb R} \int_0^T\!\rmd t \int\!\rmd y\, \big| \partial_x\partial_y K(x,y,t) \big|\le M\sqrt {T}\;.
\end{equation}
\end{lemma}
\begin{proof}
It follows simply by integration of  the expression \eqref{defK} for the kernel $K$, and estimation of each term with the aid of \eqref{Ki}, \eqref{Kt}, \eqref{ddt}, and \eqref{dyxk}. Recall the definition of $t_0$ in Proposition \ref{BKT}, and observe indeed that, from \eqref{defK}, \eqref{Ki}, and \eqref{Kt}, 
\[
\begin{split}
\int_0^{T\land t_0}\!\rmd t \int\!\rmd y\, \big|\partial_x\partial_y K(x,y,t)\big|
& \le C \int_0^{T\land t_0}\!\rmd t\int\!\rmd y\, \big(t^{-\frac 34} +t^{-\frac 12}\big)
\rme^{-t^{-\frac 14}|x-y|} \\ & \le C \big((T\land t_0)^{\frac 12}+(T\land t_0)^{\frac 34}\big)\;.
\end{split}
\]
If $T>t_0$, from \eqref{defK} and \eqref{ddt}, 
\[
\int^T_{T\land t_0}\!\rmd t \int\!\rmd y\, \big|\partial_x\partial_y k(x,y,t)\big| \le C\int^T_{T\land t_0}\!\rmd t\, \frac{1}{t} \int\!\rmd y\,\rme^{-\mu t^{-\frac 12}|x-y|} \le C T^{\frac 12}\;, 
\]
and we are left with the estimation of the term containing $K^\ast$, which can be done by integration of the five terms in \eqref{dyxk}. Call $I_1,\ldots,I_5$ the resulting integrals. To conclude the proof, it is easy to see that 
\begin{align}
\nn | I_1+I_2|&\le C \int^T_{T\land t_0}\!\rmd t\,t^{-\frac 12} \le C T^{\frac 12}\;, \qquad |I_3| \le C \int^T_{T\land t_0}\!\rmd t\, x\,t^{-\frac 32} \rme^{-\frac{x^2}{8t}} \le C T^{\frac 12}\;, \\ \nn |I_4+I_5| & \le C \int^T_{T\land t_0}\!\rmd t\, \frac{1}{t} \le C T^{\frac 12}\;.
\end{align} 
The lemma is thus proved.
\end{proof}

\begin{proposition}
\label{gapprox}
There exists a time $T_\epsilon (\omega)$ such that the equation \eqref{sche} with initial condition $\bam $ has a unique continuous bounded solution $u(x,t)$ for $t\le T_\epsilon$. Moreover, given $T>0$, $\gamma<\frac23$, and $\xi \in \big(0,\frac 12 -\frac {3\gamma}{4}\big)$, there exists a set $B_\epsilon(T,\xi)\in \mf F$ such that
\[
\bb P \{B_\epsilon(T,\xi)\} \to 1 \mbox{ as } \epsilon \to 0\;,
\]
and for $\omega \in B_\epsilon(T,\xi)$ and $\epsilon$ sufficiently small, $u(x,t)$
satisfies,
\begin{align}
\label{gapp} &\sup_{(x,t)\in \mc T_\epsilon} |u(x,t)-\bam(x)|<\epsilon^{\frac12 -\frac\gamma 4-\frac \xi 3}\;, \\ \label{gapp2} & \sup_{(x,t)\in \mc T_\epsilon}
\big|u(x,t)-\bam(x)-\sqrt \epsilon \,H(x,t)\big|<\epsilon^{1-\gamma-\xi}\;,
\end{align}
where the set $\mc T_\epsilon$ is defined in \eqref{defmct}. In particular, $\bb P \{ T_\epsilon > \epsilon^{-\gamma}T\} \to 1$ as $\epsilon \to 0$.
\end{proposition}

\begin{proof}
Recall that the equation \eqref{sche} has to be understood in the sense of its integral version \eqref{freeie}.
We already know from Proposition \ref{locuye} that, for each $\omega \in \Omega$, there exists a positive time  $T_\epsilon(\omega)$ such that the equation \eqref{freeie} admits a unique continuous solution $u$ for $t\le T_\epsilon(\omega)$. From \eqref{i}, we also know that $\|\sqrt \epsilon Y\|_{\infty,t}<\infty$ if $t\le \epsilon^{-4}\,T$. We will exhibit a set $ B_\epsilon(T,\xi) \in \mf F$  such that, if $\omega \in  B_\epsilon(T,\xi)$, then the corresponding solution $u$ with initial condition $\bam$ satisfies \eqref{gapp}, which in particular implies (from  the proof of Proposition \ref{locuye} for this particular initial condition and $\omega$) that we can take  $T_\epsilon(\omega)>\epsilon^{-\gamma}\,T$, for $\gamma <\frac 23$. To complete the proof,     we need then to show the existence of the set $B_\epsilon(T,\xi)$, whose probability goes to $1$ as $\epsilon \to 0$, and such that \eqref{gapp} and \eqref{gapp2} hold on $B_\epsilon(T,\xi)$. This will be done for the solution to the integral equation \eqref{linie} with $f=0$, which is equivalent to \eqref{freeie} with $u_0=\bam$.

In the sequel we shall use the rescaled variables $(x,\epsilon^{-\gamma}t)$, with $(x,t)\in \mc T = \{(x,t)\colon x\in \bb R,\; t\in [0,T]\}$, to parametrize the elements in $\mc T_\epsilon$. Define the set 
\[
B_\epsilon(T,\xi) = \Big\{\omega\colon \sup_{(x,t)\in\mc T} \epsilon^{\frac{\gamma}{4}+\frac \xi 3} |H(x,\epsilon^{-\gamma}t)|\le \frac 12\Big\}\;.
\]
From \eqref{iih}, $\bb P\big(B_\epsilon(T,\xi)\big)\to 1$ as $\epsilon \to 0$. To prove that \eqref{gapp} holds on this set, observe that a simple time scaling in equation \eqref{linie} gives, after \eqref{defH}, 
\begin{equation}
\label{tesc}
\begin{split}
v(x,\epsilon^{-\gamma}t) = & - \epsilon^{-\gamma}\,\int_0^t\!\rmd s \int\!\rmd y\, \partial_x\partial_y K(x,y,\epsilon^{-\gamma}(t-s)) \big(3\bam v^2+v^3\big) (y,\epsilon^{-\gamma}s) \\ & + \sqrt \epsilon\, H(x,\epsilon^{-\gamma}t)\;.
\end{split}
\end{equation}
In terms of $D_\epsilon(x,t)=\epsilon^{-\frac12+\frac\gamma 4+\frac \xi 3} v(x,\epsilon^{-\gamma}t)$, the previous equation reads,
\[
\begin{split}
D_\epsilon(x,t) & = \epsilon^{\frac 12 -\frac{5\gamma}{ 4}-\frac \xi 3} \int_0^t\!\rmd s \int\!\rmd y\, \partial_x\partial_y K(x,y,\epsilon^{-\gamma}(t-s)) \\ & \quad \times  \big(3\bam D^2_\epsilon + \epsilon^{\frac 12-\frac{\gamma}{4}-\frac \xi 3} D^3_\epsilon \big)(y,s) + \epsilon^{\frac\gamma 4+\frac \xi 3}\, H(x,\epsilon^{-\gamma}t)\;.
\end{split}
\]
Recall that $D_\epsilon(x,0)=0$ and define next the stopping time, 
\[
T^* =\inf \{t>0\colon \|D_\epsilon(\cdot,t)\|_\infty \ge 1\}\;.
\]
From \eqref{cddk} we obtain then that for any $\tau\le T^*$,
\[
\|D_\epsilon(\cdot,\tau)\|_\infty \le \epsilon^{\frac 12 -\frac{3\gamma}{4}-\frac \xi 3}C \tau^{\frac 12} \big(3\,+\epsilon^{\frac12-\frac \gamma 4 -\frac \xi 3} \big) + \sup_{x\in \bb R,t\le \tau} \epsilon^{\frac\gamma 4+\frac \xi 3} |H(x,\epsilon^{-\gamma}t)|\;.
\]
This inequality implies that $T^* > T$ for $\omega \in B_\epsilon(T,\xi)$, for otherwise, evaluating at $T^*$ we would get, 
\[
 1\le C T^{\frac 12} \epsilon^{\frac 12 -\frac{3\gamma}{4}-\frac \xi 3}+\frac 12\;, 
\] 
which cannot be true for sufficiently small $\epsilon$, under the assumptions on $\gamma$ and $\xi$.  But $T^* > T$ is precisely \eqref{gapp}. 

To prove \eqref{gapp2}, we notice that on the set $B_\epsilon(T,\xi)$, from \eqref{tesc} and since \eqref{gapp} holds on that set,  
\[
\begin{split}
& \sup_{(x,t)\in\mc T} |v(x,\epsilon^{-\gamma}t) - \sqrt \epsilon H(x,\epsilon^{-\gamma}t)| \\ & \qquad \le 4 \epsilon^{-\gamma +1-\frac \gamma 2 -\frac{2 \xi}{3}} \sup_{(x,t)\in\mc T} \Big|\int_0^t\! \rmd s \int\!\rmd y\, \partial_x\partial_y K(x,y,\epsilon^{-\gamma}(t-s)) \Big| \\ & \qquad \le C T^{\frac 12} \epsilon^{1-\gamma -\frac{2 \xi}{ 3}}\;, 
\end{split}
\]
which implies \eqref{gapp2}, for $\epsilon$ sufficiently small. 
\end{proof}

The previous result shows that, with probability going to $1$, the solution $u$ to the stochastic Cahn Hilliard equation with initial condition $\bam$ remains close to it for times of the order of $\epsilon^{-\gamma}$ if $\gamma<\frac 23$. It also gives an idea about the fluctuations. To make it precise, we recall the notion of center of a front, already considered in \cite{BDP,BBDP,BBB}  to study the front fluctuations for the Allen-Cahn and phase field equations. Motivation and properties can be found in those articles and some of the references therein. 
 
Recall the definition \eqref{mcm} of $\mc M $, and  consider
\[
\mc M_\delta = \big\{m\in C^0(\bb R)\colon \mbox{dist}(m,\mc M) = \inf_{x_0\in \bb R}\|m-\bam_{x_0}\| \le \delta \big\}\;,
\]
and for $m\in \mc M_\delta$ define the \emph{center} of $m$ as the real number $\xi$ such that
\[
\langle m-\bam_\xi,\bam_\xi'\rangle = 0\;.
\]
The following result is proved in \cite{BDP}.

\begin{lemma}
\label{32}
There exists $\delta_0>0$ such that, if $\delta \le \delta_0$ and $m\in \mc M_\delta$, then $m$ has a unique center $\zeta\in \bb R$. If $x_0$ is such that  $\|m-\bam_{x_0}\|_{\infty}<\delta$, then there exists a constant $C$ depending only on $\delta$ such that 
\begin{align}
\nonumber
&{\rm (i)}\quad  |x-\zeta|\le C\, \|m-\bam_{x_0}\|_\infty\;, \\ \nonumber & {\rm (ii)}\quad \zeta=x_0-\frac 34 \langle m-\bam_{x_0},\bam_{x_0}'\rangle - \frac {9}{16}\langle m-\bam_{x_0},\bam_{x_0}'\rangle \langle m-\bam_{x_0},\bam_{x_0}''\rangle\,+R\;,
\end{align}
where the remainder $R\le C\,\|m-\bam_{x_0}\|_\infty^3$. 
\end{lemma}

We can now prove the following,
\begin{lemma}
\label{33}
Let $u$ be the solution to  \eqref{sche} with initial condition $\bam$, $T>0$, and $\gamma \in\big(0,\frac 23\big)$. Then, on a set whose probability goes to $1$ as $\epsilon\to 0$, $u(\cdot,t)$ has a unique center $z_\epsilon(t)$ for any $t\le \epsilon^{-\gamma}T$ and any $\epsilon$ sufficiently small. It satisfies, \begin{equation}
\label{apzeta}
z_\epsilon(t)=-\frac 34 \langle u(\cdot,t)-\bam,\bam'\rangle + R_\epsilon\;,
\end{equation}
where, for any given $\xi>0$, $\sup_{t\le \epsilon ^{-\gamma}T}|R_\epsilon| \le C\epsilon^{1-\frac{\gamma}2-\xi}$.
\end{lemma}

\begin{proof}
From \eqref{gapp}, in the set $B_\epsilon(T,\xi)$ the solution $u(\cdot,t)$ has a unique center $z_\epsilon(t)$ for any $t\le \epsilon^{-\gamma}T$. Moreover, by item (ii) of Lemma \ref{32} with $x_0=0$ and \eqref{gapp}, the center $z_\epsilon(t)$ satisfies \eqref{apzeta}.
\end{proof}

\section{The process $H$}
\label{sec:6}

We proceed now to establish some properties of the process $H$ defined in \eqref{defH}, that will be useful to study the fluctuations of the center $z_\epsilon(t)$ of $u$, as suggested by \eqref{apzeta} and \eqref{gapp2}. 

\begin{lemma}
The process $H(x,t)$ may be decomposed as 
\begin{equation}
\label{decH}
H(x,t)=H_1(x,t)+H_2(x,t)\;,
\end{equation}
where
\begin{align}
\nn {\rm(i)}\quad  &H_2(x,t)=-\frac 12 \bam '(x) \int_0^t \int  \frac{\rme^{-\frac{y^2}{8(t-s)}}}{\sqrt{2\pi (t-s)}} \sign(y) a_\epsilon(y)\, \rmd W_{y,s}\;.
\\
\nn
{\rm(ii)}\quad & \mbox{For each $\xi>0$, $T>0$, and $\gamma<\frac 23$, the process $H_1(x,t)$ satisfies},
\\
\label{coth1}
& \qquad \qquad \qquad \lim_{\epsilon \to 0}\bb P \big(\sup _{(x,t)\in \mc T_\epsilon}|H_1(x,t)|>\epsilon^{-\xi}\big) = 0\;.
\end{align}
\end{lemma}

\begin{proof}
From \eqref{defH} and \eqref{defK} we may write,
\begin{align}
\nn
H(x,t) & = \id _{\{t\le t_0\}}\int_0^t \int \partial_x(K_\infty+\tilde K) (x,y,t-s)a_\epsilon(y)\,\rmd W_{y,s} \\ \nn & \quad + \id _{\{t> t_0\}}\Big(\int^{t-t_0}_0 \int (T_1+T_3+T_4)(x,y,t-s) a_\epsilon(y)\,\rmd W_{y,s} \\ \nn & \qquad\qquad + \int^{t-t_0}_0 \int \partial_x k (x,y,t-s) a_\epsilon(y)\,\rmd W_{y,s} \\ \nn & \qquad\qquad + \int_{t-t_0}^t \int \partial_x (K_\infty+\tilde K) (x,y,t-s) a_\epsilon(y)\,\rmd W_{y,s} \\ \nn & \qquad\qquad + \bam'(x) \int_{t-t_0}^t \frac{1}{2\sqrt{2\pi (t-s)}}\sign (y) \rme^{-\frac{y^2}{8(t-s)}}a_\epsilon(y)\,\rmd W_{y,s}\Big) \\ \nn & \quad + H_2(x,t)\;.
\end{align}
The term $T_j$ above ($j=1,3,4$) is defined as the first, third and fourth term in the kernel $\partial_x K^*$ in \eqref{ks}, respectively. Call $H_1(x,t)$ the process given by all the terms but the last one on the right hand side above and let us  show that $H_1$ so defined  satisfies (ii), thus concluding the proof. 

From the proof of \eqref{varH} one can see that only term  of order $t^{\frac 12}$ is  precisely the one coming from that part of $\partial_x K^*$ which is now in $H_2(x,t)$, see \eqref{dxKs}. Therefore, 
\begin{equation}
\label{svarh1}
\bb E H_1(x,t)^2\le C\big(\log(1+t)+1\big)\;,
\end{equation}
for some constant $C$ independent of $x,t$. Also, from \eqref{e1}-\eqref{e6}, it is clear that the estimate \eqref{ex} is valid for $H_1$ as well: given $\delta>0$, for any $\epsilon$ sufficiently small,
\begin{equation}
\label{varh1}
R>\epsilon^{-\frac{11\gamma}{10}} \delta^{-2}+1 \quad \Longrightarrow \quad  \sup_{|x|\ge R +\epsilon^{-\beta},\, t\le \epsilon^{-\gamma}T} \bb E H_1(x,t)^2 \le \delta ^2\;.
\end{equation} 
On the other hand, given $h>0$,
\begin{equation}
\label{dh2x}
\begin{split}
& \bb E \big(H_2(x+h,t)- H_2(x,t)\big)^2 \\ & \quad = \frac 14 \big({\bam '}(x+h)-{\bam '}(x)\big)^2\,\int_0^t\! \rmd s\int\!\rmd y\, \frac{\rme^{-\frac{y^2}{4s}}}{2\pi s} a_\epsilon(y)^2 \le C ht^{\frac 12}\;,
\end{split}
\end{equation}
\begin{equation}
\label{dh2t}
\begin{split}
& \bb E \big(H_2(x,t+h)- H_2(x,t)\big)^2 \\ & \quad = \frac{{\bam'(x)^2}}{4}\int^{t+h}_t\!\rmd s \int\!\rmd y\, \frac{\rme^{-\frac{y^2}{4(t+h-s)}}}{2\pi (t+h-s)} a_\epsilon(y)^2 \\ & \qquad + \frac{{\bam'(x)^2}}{4} \int^{t}_0\!\rmd s \int\!\rmd y\, \Big(\frac{\rme^{-\frac{y^2}{8(t+h-s)}}}{\sqrt{2\pi (t+h-s)}} -\frac{\rme^{-\frac{y^2}{8(t-s)}}}{\sqrt{2\pi (t-s)}}\Big)^2 a_\epsilon(y)^2  \\ & \quad \le C \Big(\int^h_0\frac{\rmd s}{\sqrt s} +\int_0^t\!\rmd s\, \frac{1}{\sqrt {2(s+h)}} + \frac{1}{\sqrt {2s}}-\frac{2}{\sqrt {2s +h}}\Big) \le C h^{\frac 12}\;.
\end{split}
\end{equation}
From \eqref{decH}, \eqref{hxH}, and \eqref{dh2x} we have,
\begin{equation}
\label{dh1x}
\bb E \big(H_1(x+h,t)- H_1(x,t)\big)^2 \le C h\big(t^{\frac 12} \big(1+ h^3+ h|x|^2\big) + \log\big(1+h^{-1}t^{\frac 14}\big)\big)\;.
\end{equation}
Analogously, from \eqref{htH} and \eqref{dh2t}, 
\begin{equation}
\label{dh1t}
\bb E \big(H_1(x,t+h)- H_1(x,t)\big)^2 \le C \big(h^{\frac 14}+h^{\frac 32} + (h+h^{\frac12}) t^{\frac 12}\big)\;.
\end{equation}
With the aid of \eqref{dh1x}, \eqref{dh1t}, and \eqref{varh1}, proceeding as in the proof of \eqref{esupH}, it follows that, for any $\xi>0$ and sufficiently small $\epsilon$, 
\[
\bb E \big(\sup_{x \in\mc T_\epsilon} | H_1(x,t)|\big) < \epsilon ^{-\frac\xi 5}\;.
\]
Finally, from  \eqref{svarh1} and Borell's inequality (as in \eqref{bineq}), (ii) follows.
\end{proof}

We consider next the asymptotics for the scaled process $\epsilon^{\frac\gamma 4}H(x,\epsilon^{-\gamma}t)$, which yields the  leading term for $u(x,\epsilon^{-\gamma}t) -\bam (x)$ (see Proposition \ref{gapprox}). 

\begin{proposition}
\label{asymh}
For any $x,x'\in \bb R$ and $t>t'\ge 0$,
\begin{equation} 
\label{fdd}
\lim_{\epsilon\to 0}\epsilon ^{\frac{\gamma}{2}} \bb E\big(H(x,\epsilon^{-\gamma}t)H(x',\epsilon^{-\gamma}t')\big)  = \frac{\bam'(x)\bam'(x')}{2\,\sqrt{2 \pi}} \big(\sqrt{t+t'}-\sqrt{t-t'}\big)\;.
\end{equation}
\end{proposition}

\begin{proof}
Changing variables in \eqref{covH} we obtain, 
\begin{equation}
\label{covHe}
\begin{split}
& \bb E\big(H(x,\epsilon^{-\gamma}t)H(x',\epsilon^{-\gamma}t')\big) \\ & \quad =\int^{\epsilon^{-\gamma}t'}_0\!\rmd s \int\!\rmd y\, \partial _x K(x,y,s) \partial _{x'} K\big(x',y,s+\epsilon^{-\gamma}(t-t') \big)a_{\epsilon}(y)^2  \;.
\end{split}
\end{equation}
Recall formula \eqref{defK} for $K$, observe that, as $\epsilon \to 0$, we only need to consider $\epsilon ^{-\gamma}t'>t_0$, split the time  integral above according $s\le t_0$ or not, and call $\mc E_1$ and $\mc E_2$ the resulting terms. Let us consider first the latter. From \eqref{defK},  
\begin{equation}
\label{E2}
\begin{split}
& \mc E_2 = \int^{\epsilon^{-\gamma}t'}_{t_0}\!\rmd s \int\!\rmd y\, \big(\partial _x  K^\ast (x,y,s)+\partial _xk(x,y,s)\big) a_{\epsilon}(y)^2  \\ & \qquad \times \big(\partial _{x'} K^\ast (x',y,s+(t-t')\,\epsilon^{-\gamma})+\partial _{x'} k(x',y,s+(t-t')\,\epsilon^{-\gamma})\big)\;.
\end{split}
\end{equation}
Substituting  $\partial _{x}K^\ast$ and $\partial_{x'}K^\ast$ by the expressions in \eqref{ks}, we write  $\mc E_2$ as a sum of integrals. Let us single out the integral corresponding to the product of each second term in the right hand side of \eqref{ks}, and denote it by $\mc E_{22}$,
\begin{equation}
\label{e22}
\begin{split}
\mc E_{22} & = \frac{\bam'(x) \bam'(x') }4\int^{\epsilon^{-\gamma}t'}_{t_0}\!\rmd s \int\!\rmd y\,\frac{\rme^{-\frac{y^2}{8s}}\,\rme^{-\frac{y^2}{8(s+\epsilon^{-\gamma}(t-t'))}}a_{\epsilon}(y)^2}{\sqrt{2\pi s}\,\sqrt{2\pi (s+(t-t')\,\epsilon^{-\gamma})}}\;.
\end{split}
\end{equation}
To conclude the proof, we will show that $\mathcal E_{22}$ is the only term that contributes to the limit in \eqref{fdd}, that is,
\begin{itemize}
\item[(i)]$\lim_{\epsilon\to 0}\epsilon^{\frac{\gamma}{2}}\mc E_{22}=\frac{1}{2\,\sqrt{2 \pi}}\bam'(x)\bam'(x') \big(\sqrt{t+t'}-\sqrt{t-t'}\big)$;
\item[(ii)]$\lim_{\epsilon\to 0}\epsilon^{\frac{\gamma}{2}}\big(\mc E_{2}-\mc E_{22}\big)=0$;
\item[(iii)]$\lim_{\epsilon\to 0}\epsilon^{\frac{\gamma}{2}}\mc E_{1}=0$.
\end{itemize}

\noindent
{\it Proof of} (i). Let us compute the integral in \eqref{e22}, but taking $a_{\epsilon} =1$,
\begin{equation}
\label{e22a1}
\begin{split}
\int^{\epsilon^{-\gamma}t'}_{t_0}\!\rmd s & \int\!\rmd y\,\frac{\rme^{-\frac{y^2}{8s}}\,\rme^{-\frac{y^2}{8(s+\epsilon^{-\gamma}(t-t'))}}}{\sqrt{2\pi s} \sqrt{2\pi (s+ \epsilon^{-\gamma} (t-t') )}}
\\ & =  2 \int^{t'\epsilon^{-\gamma}}_{t_0}\! \frac{\rmd s}{\sqrt{2\pi(2s+\epsilon^{-\gamma}(t-t'))}}
\\ & = \frac{2}{\sqrt{2\pi}}\Big(\sqrt{\epsilon^{-\gamma}(t+t')}-\sqrt{2t_0+\epsilon^{-\gamma}(t-t')}\Big)\;.
\end{split}
\end{equation}
Then, (i) follows once we show that 
\begin{equation}
\label{h2beta}
\lim_{\epsilon\to 0}\epsilon^{\frac{\gamma}{2}}\int^{\epsilon^{-\gamma}t'}_{t_0}\!\rmd s \int\!\rmd y\, \frac{\rme^{-\frac{y^2}{8s}}\,\rme^{-\frac{y^2}{8(s+\epsilon^{-\gamma}(t-t'))}} (1-a_{\epsilon}(y)^2)}{\sqrt{2\pi s}\,\sqrt{2\pi (s+(t-t')\epsilon^{-\gamma})}}= 0\;.
\end{equation}
To do this, we split the spatial integral in \eqref{h2beta} according to $|y|\le \epsilon ^{-\frac \beta 2}$ or not. For the first case, from the properties of $a$, we know that, given $\eta>0$, $|1-a_{\epsilon}(y)^2| <\eta$ for $\epsilon$ sufficiently small. Computing the integral integral as above we get,
\begin{equation}
\label{lbeta}
\epsilon^{\frac \gamma 2} \int^{\epsilon^{-\gamma}t'}_{t_0}\!\rmd s \int_{|y|\le\epsilon ^{-\frac \beta 2}}\!\rmd y\,\frac{\rme^{-\frac{y^2}{8s}}\,\rme^{-\frac{y^2}{8(s+\epsilon^{-\gamma}(t-t'))}}(1-a_{\epsilon}(y)^2) }{\sqrt{2\pi s}\,\sqrt{2\pi (s+(t-t')\epsilon^{-\gamma})}} \le C\eta\;.
\end{equation}
In the other case, we have,
\begin{equation}
\label{gbeta}
\begin{split}
& \epsilon^{\frac \gamma 2} \int^{\epsilon^{-\gamma}t'}_{t_0}\!\rmd s \int_{|y|>\epsilon ^{-\frac \beta 2}}\!\rmd y\,\frac{\rme^{-\frac{y^2}{8s}}\,\rme^{-\frac{y^2}{8(s+\epsilon^{-\gamma}(t-t'))}}(1-a_{\epsilon}(y)^2) }{\sqrt{2\pi s}\,\sqrt{2\pi (s+(t-t')\epsilon^{-\gamma})}} \\ & \quad \le C\epsilon^{\frac \gamma 2} \int^{\epsilon^{-\gamma}t'}_{t_0}\!\rmd s \, \frac{\rme^{-\frac{\epsilon^{-\beta}}{8s}}}{\sqrt{2\pi s}} \le C \rme^{-\epsilon^{\frac{-\beta+\gamma}{8}}} \longrightarrow 0 \quad \mbox{as } \epsilon\to 0\;,
\end{split}
\end{equation}
since $\gamma<\beta$, and (i) follows. 

\smallskip
\noindent {\it Proof of} (ii). We need to consider all the terms in \eqref{E2} except $\mc E_{22}$. From H\"older's inequality and \eqref{sk} we have, 
\begin{equation}
\label{dkk}
\begin{split}
& \Big|\int^{\epsilon^{-\gamma}t'}_{t_0}\!\rmd s \int\!\rmd y\,\partial _x k(x,y,s)\partial _{x'} k(x',y,s+\epsilon^{-\gamma}(t-t'))a_{\epsilon}(y)^2 \Big| \\ & \qquad \le \Big(\int^{\epsilon^{-\gamma}t'}_{t_0}\!\rmd s \int\!\rmd y\,\big(\partial _xk(x,y,s)\big)^2\Big)^{\frac 12 } \\ & \qquad\quad \times \Big(\int^{\epsilon^{-\gamma}t}_{t_0+\epsilon^{-\gamma}(t-t')}\!\rmd s \int\!\rmd y\,a_{\epsilon}(y)^2|\partial _{x'} k(x',y,s)|^2\Big)^{\frac 12}\le C\;.
\end{split}
\end{equation}
Analogously, from \eqref{sk} and \eqref{dxKs},
\begin{equation}
\label{dxkKs}
\begin{split}
& \Big|\int^{\epsilon^{-\gamma}t'}_{t_0}\!\rmd s \int\!\rmd y\,\partial_x K^\ast(x,y,s) \partial_{x'}k(x',y,s+\epsilon^{-\gamma}(t-t'))a_{\epsilon}(y)^2 \Big| \\ & \qquad \le \Big(\Big|\int^{\epsilon^{-\gamma}t'}_{t_0}\!\rmd s \int\!\rmd y\,\big(\partial _x K^\ast(x,y,s)\big)^2a_{\epsilon}(y)^2 \Big)^{\frac 12} \\ & \quad \qquad\qquad
 \times \Big(\int^{\epsilon^{-\gamma}t}_{t_0+\epsilon^{-\gamma}(t-t')}\!\rmd s \int\!\rmd y\, \big(\partial _{x'} k(x',y,s)\big)^2a_{\epsilon}(y)^2 \Big)^{\frac 12} \\ & \qquad \le C \big( \log (\epsilon^{-\gamma}t)+\sqrt{\epsilon^{-\gamma}t}\big)^{\frac 12}\;.
\end{split}
\end{equation}
Finally, to estimate the terms coming from $\partial _x K^\ast\partial _{x'} K^\ast$ in \eqref{E2},
we observe that, from \eqref{ks}, H\"older's inequality and \eqref{dxKs},
\begin{equation}
\label{tf1}
\begin{split}
&\Big|\int^{\epsilon^{-\gamma}t'}_{t_0}\!\rmd s \int\!\rmd y\,\partial _x K^\ast(x,y,s) \frac{ \bam'(x')\varphi(y)a_{\epsilon}(y)^2}{2\sqrt{2\pi (s+\epsilon^{-\gamma}(t-t'))}}\Big|
\\ &\qquad \qquad\le C \big( \log (t'\epsilon^{-\gamma})+(\epsilon^{-\gamma}t')^{\frac 12}\big)^{\frac 12}
\big(\log (\epsilon^{-\gamma}t)\big)^{\frac 12}\;,
\end{split}
\end{equation}
\begin{equation}
\label{tf2}
\begin{split}
&\Big|\int^{\epsilon^{-\gamma}t'}_{t_0}\!\rmd s \int\!\rmd y\,\partial _x K^\ast(x,y,s)\,\frac{x'\,\varphi(y)\rme^{-\frac{{x'}^2}{8(s+\epsilon^{-\gamma}(t-t'))}}a_{\epsilon}(y)^2 }{8\,\sqrt {2\pi} (s+\epsilon^{-\gamma}(t-t'))^{\frac 32}}\Big| \\ &\qquad \qquad\le C \big( \log (\epsilon^{-\gamma}t)+(\epsilon^{-\gamma}t)^{\frac 12}\big)^{\frac 12}\;,
\end{split}
\end{equation}
\begin{equation}
\label{tf3}
\begin{split}
&\Big|\int^{\epsilon^{-\gamma}t'}_{t_0}\!\rmd s \int\!\rmd y\,\partial _x K^\ast(x,y,s) \,\frac{(x'+y)\rme^{-\frac{(x'+y)^2}{8(s+\epsilon^{-\gamma}(t-t'))}}a_{\epsilon}(y)^2 }{4\,\sqrt {2\pi} (s+\epsilon^{-\gamma}(t-t'))^{\frac 32}} \Big|\\ & \qquad \qquad\le C \big( \log (\epsilon^{-\gamma}t)+(\epsilon^{-\gamma}t)^{\frac 12}\big)^{\frac 12}\;.
\end{split}
\end{equation}
Observe that \eqref{dkk}, \eqref {dxkKs}, \eqref{tf1}, \eqref{tf2}, and \eqref {tf3} go to zero when multiplied by $\epsilon ^{\frac \gamma 2}$, as well as the remaining terms (which are similarly estimated), and (ii) follows. 

\smallskip
\noindent{\it Proof of} (iii). From H\"older's inequality, 
\[
\begin{split}
 |\mc E_1| & = \Big| \int^{t_0}_0\!\rmd s \int\!\rmd y\, \partial _x K(x,y,s) \partial _{x'} K\big(x',y,s+\epsilon^{-\gamma}(t-t')\big) a_{\epsilon}(y)^2 \Big| \\ & \le \Big(\int^{t_0}_0\!\rmd s \int\!\rmd y\,  (\partial _x K(x,y,s))^2a_{\epsilon}(y)^2\Big)^\frac 12  \\ & \qquad \times  \Big(\int^{t_0}_0\!\rmd s \int\!\rmd y\,  \big(\partial _{x'}K(x',y,s+\epsilon^{-\gamma}(t-t'))\big)^2a_{\epsilon}(y)^2\Big)^\frac 12\;.
\end{split}
\]
The first factor on the right hand side above is bounded by a constant, as can be seen from \eqref{defK}, \eqref{eki} and \eqref{tildek}. For the last factor, from \eqref{defK} we have that, if $s\ge t_0$,  
\[
|\partial _{x'} K(x',y,s)|^2 \le 2 |\partial _{x'} K^\ast(x',y,s)|^2+2 |\partial _{x'} k(x',y,s)|^2\;.
\]
From \eqref{sk},  we obtain, 
\[
\int_{\epsilon^{-\gamma}(t-t')}^{\epsilon^{-\gamma}(t-t')+t_0}\!\rmd s \int\!\rmd y\,
 \big(\partial _{x'} k(x',y,s)\big)^2a_{\epsilon}(y)^2 \le C\;.
\]
Finally, we need to  integrate $K^\ast$. As in \eqref{dxKs},  from \eqref{ks}, 
\[
\begin{split}
& \int_{\epsilon^{-\gamma}(t-t')}^{\epsilon^{-\gamma}(t-t')+t_0}\!\rmd s \int\!\rmd y\, \big(\partial _{x'} K^\ast(x',y,s)\big)^2 a_{\epsilon}(y)^2 \\ & \qquad \qquad \qquad \le C 
 \int^{\epsilon_{-\gamma}(t-t')}_{\epsilon^{-\gamma}(t-t')+t_0}\!\rmd s\, \Big(\frac 1s+\frac {1}{s^{\frac 12}}+\frac{1}{s^{\frac 32}} \Big) \le C\;.
\end{split} 
\]
From the previous estimates, $|\mc E_1|\le C$, which implies (ii), and  concludes the proof. 
\end{proof}

The previous proposition implies that, as $\epsilon \to 0$, the scaled process $\epsilon^{\frac \gamma 4}H(x,t\epsilon^{-\gamma })$ converges in the sense of finite dimensional distributions  
to the process $(8\pi)^{-\frac 14}\,\bam'(x) r(t)$, where $r(t) $ is the one dimensional Gaussian process with covariance \eqref{dcovr}. Recall now the decomposition of $H$ given in  Lemma \ref{decH}, and denote by $h_2$ the temporal part of $H_2$, i.e.,
\begin{equation}
\label{defh2}
h_2(t)\,=\, \int_0^t \int\! \frac{\rme^{-\frac{y^2}{8(t-s)}}}{2\sqrt{2\pi (t-s)}} \sign (y) \rme^{-\frac{y^2}{8(t-s)}} a_\epsilon(y) \,\rmd W_{y,s}\;.
\end{equation}
We show next that $h_2$,  when suitably scaled, converges weakly in $C(\bb R_+)$ (equipped with the topology of uniform convergence in compacts) to the process $r$. 

\begin{lemma}
\label{hwcr}
As $\epsilon \to 0$, the real process
\begin{equation}
\label {defheps}
h^{(\epsilon)}(t) = \epsilon^{\frac\gamma 4}h_2(\epsilon^{-\gamma}t)
\end{equation}
converges weakly in $C(\bb R_+)$ to $(8\pi)^{\frac 14} r(t)$.
\end{lemma}
 
\begin{proof}
From \eqref{e22a1} and \eqref{h2beta} with $t_0=0$ , it follows that 
\[
\lim_{\epsilon\to 0} \bb E\big(h^{(\epsilon)}(t)\,h^{(\epsilon)}(t')\big) =  (8\pi)^{-\frac 12}
\big( \sqrt {t+t'}- \sqrt {t-t'}\big)\;,
\] 
which implies $h^{(\epsilon)}(t) \longrightarrow r(t)$, in the sense of finite dimensional distributions.  From \eqref{dh2t} we also know that 
\[
\begin{split}
\bb E\big(h^{(\epsilon)}(t+h)-h^{(\epsilon)}(t) \big)^2 & = \epsilon^{\frac\gamma 2} \bb E\big(h_2(\epsilon^{-\gamma}(t+h))-h_2(\epsilon^{-\gamma}\,t)\big)^2 \\ & \le C \epsilon^{\frac \gamma 2}\,(h\, \epsilon^{-\gamma})^{\frac 12} =  C h^{\frac 12} \;,
\end{split}
\]
which, together with the fact  $h^{(\epsilon)}(0)=0$ implies that the corresponding family of laws is tight, and then the weak limit has to be $r$. 
\end{proof}

\noindent {\bf Proof of Theorem \ref{t1}.}
The statement concerning the uniqueness and existence of a solution $u$ to the Cahn Hilliard equation \eqref{sche}, or, equivalently  (as already discussed),  to the integral equation \eqref{linie} follows from  Proposition \ref{gapprox}, just by taking $B(\epsilon,\gamma,T)=B_\epsilon(T,\xi)$ with any $\xi$ small enough. 

\smallskip 
To prove item (i), given $\eta>0$ we fix $\xi \in \big(0,\frac 12 -\frac {3\gamma}{4}\big)$ such that $2\xi <\eta$ and we let $B_\epsilon(T,\xi)$ be as in Proposition \ref{gapprox}. Consider $\delta _0$ as in Lemma \ref{32} and define the process $\zeta_{\epsilon}(t)=z_\epsilon(t\land \tau)$, where 
\[
\tau=\inf \{t>0: u(.,t)\notin \mc M_{\delta_0}\}
\]
and $z_\epsilon (t)$ is the center of $u(\cdot,t)$ (which is well defined, as  follows from Lemma \ref{33} ). The process $\zeta_\epsilon(t)$ is clearly continuous and adapted to $\mf F_t$; note also that $B_\epsilon(T,\xi) \subset  \{\tau>\epsilon^{-\gamma}T\}$ for any $\epsilon$ sufficiently small. Let now,
\[
\nu_\epsilon (t)=-\frac 34 \langle u(\cdot,t)-\bam,\bam'\rangle\;,
\]
which is an approximation to $\zeta_\epsilon(t)$. Indeed, from \eqref{apzeta},  on the set  
$B_\epsilon(T,\xi)$, $\sup_{t\le \epsilon ^{-\gamma}T}|\zeta_\epsilon(t)-\nu_\epsilon(t)|\le \epsilon^{1-\frac{\gamma}{2}-\xi}$, which implies, for any $\epsilon$ small enough, 
\begin{equation}
\label{zmn}
\sup _{t\le \epsilon^{-\gamma}T}\|\bam_{\zeta_\epsilon(t)}-\bam_{\nu_\epsilon(t)}\|_{\infty}
\le \epsilon^{1-\frac{\gamma}{2}-2\xi}\;.
\end{equation}
Observe next that, from \eqref{decH} and \eqref{defh2}, after recalling that $\langle \bam ',\bam '\rangle=\frac 43$,
\[
\nu_\epsilon(t) = - \frac 34 \langle u(\cdot,t) - \bam -\sqrt \epsilon H,\bam'\rangle -\frac 34 \langle \sqrt \epsilon  H_1,\bam'\rangle +\sqrt \epsilon h_2(t)\;.
\]
Therefore, by \eqref{gapp2} and \eqref{coth1}, there is a set $\tilde B_\epsilon(T,\xi) \subset B_\epsilon(T,\xi)$, with $\bb P(\tilde B_\epsilon(T,\xi)) \to 1$ as $\epsilon\to 0$, such that 
\[
\sup_{t\le \epsilon ^{-\gamma}T} |\nu_\epsilon(t)-\sqrt \epsilon \, h_2(t)| \le \epsilon ^{1-\gamma -\xi}+ \epsilon ^{\frac 12-\xi} \qquad \forall\, \omega\in \tilde B_\epsilon(T,\xi)\;,
\]
and then, for any $\epsilon$ small enough,
\begin{equation}
\label{deh}
\sup _{t\le \epsilon^{-\gamma}T} \|\bam_{\nu_\epsilon (t)}-\bam_{\sqrt \epsilon h_2(t)}\|_{\infty} \le \epsilon ^{1-\gamma -2\xi}+ \epsilon ^{\frac 12-2\xi} \qquad \forall\, \omega\in \tilde B_\epsilon(T,\xi)\;.
\end{equation} 

Now, by triangle inequality we may write,
\begin{equation}
\label{ti1}
\begin{split}
\|u(\cdot,t)-\bam _{\zeta_\epsilon(t)}\|_\infty & \le \|\bam_{\zeta_\epsilon(t)}-\bam_{\nu_\epsilon(t)}\|_{\infty} + \|\bam_{\nu_\epsilon (t)}-\bam_{\sqrt \epsilon h_2(t)}\|_{\infty}  \\ & \quad + \|u(\cdot,t)-\bam_{\sqrt \epsilon h_2(t)}\|_\infty\;. 
\end{split}
\end{equation}
For this last term, from \eqref{decH} and \eqref{defh2} we have in turn,
\begin{equation}
\label{ti2}
\begin{split}
& \|u(\cdot,t)-\bam_{\sqrt \epsilon h_2(t)}\|_\infty \le \|u(\cdot, t)-\bam-\sqrt \epsilon H\|_\infty + 
\|\bam +\sqrt \epsilon H-\bam_{\sqrt \epsilon h_2(t)}\|_\infty \\ & \le \|u(\cdot,t)-\bam-\sqrt \epsilon H\|_\infty + \|\bam - \sqrt \epsilon \bam'  h_2-\bam_{\sqrt \epsilon h_2(t)}\|_\infty 
+\sqrt \epsilon \|H_1\|_\infty\;.
\end{split}
\end{equation}
The first and last terms on this last line are bounded with the aid of \eqref{gapp2} and \eqref{coth1}, while for the middle one we have, for some $\theta \in \bb R$, 
\begin{equation}
\label{m2}
|\bam-\bam_{\sqrt \epsilon h_2(t)}-\bam ' \sqrt \epsilon h_2(t)| = \frac \epsilon2 |\bam''(\theta)|h_2^2(t)\;,
\end{equation} 
and, from Lemma \ref{hwcr}, we know that  
\begin{equation}
\label{if}
\lim_{\epsilon\to 0} \bb P \big(\sup_{t\le \epsilon^{-\gamma}T} \epsilon |h_2^2(t)| >\epsilon ^{1-\frac \gamma 2-2\xi}\big) = 0\;.
\end{equation}
Item (i)  follows now from \eqref{ti1}, \eqref{zmn}, \eqref{deh}, \eqref{ti2}, \eqref{m2},  and \eqref{if}.

The convergence of $X_\epsilon (t)=\epsilon ^{-\frac 12 +\frac \gamma 4}\,\zeta_\epsilon(\epsilon ^{-\gamma} t) $ follows from the above estimates and Lemma \ref{hwcr} after writing 
\[
\zeta_\epsilon(t) = (\zeta_\epsilon(t)-\nu_\epsilon(t))+(\nu_\epsilon-\sqrt \epsilon h_2(t))+\sqrt \epsilon h_2(t)\;,
\]
which proves (ii). 

\section{The one dimensional process $r$}
\label{sec:7}

In the next proposition we summarize some properties of $r$, that follow at once from the form of the covariance function.  

\begin{proposition}
The one dimensional Gaussian process with covariance given by \eqref{dcovr} satisfies,
\begin{itemize}
\item[{\rm (i)}] It has (a modification with)  continuous paths.
\item[{\rm (ii)}] It is a self similar process of order $\frac 14$, that is, for any given $a>0$,
\[
\{r(at)\}_{t\ge 0}\overset{\rm law}{=}\{a^{\frac 14} r(t)\}_{t\ge 0}\;.
\]
\end{itemize}
\end{proposition}

We also obtain several representations of the process $r$, in terms of fractional Brownian motion, the solution to the $1-D$ stochastic heat equation and a stochastic integral with respect to Brownian motion, that may be of independent interest. 

Let us recall that a two sided fractional-Brownian motion with Hurst parameter $H$, is a one dimensional Gaussian process $\nu^{(H)}(t)$ characterized by its covariance function, 
\[
\bb E\big(\nu^{(H)}(t)\nu^{(H)}(s)\big) = \frac12 \big(|t|^{2H}+|s|^{2H}-|t-s|^{2H}\big)\;.
\]
Define the odd part of $\nu^{(H)}(t)$ as usual,
\begin{equation}
\label{odd}
\nu_O^{(H)}(t) = \frac 12 \big( \nu^{(H)}(t)-\nu^{(H)}(-t)\big)\;.
\end{equation}
We refer to \cite{MV} for an introduction and properties of fractional Brownian motion. 

Next, consider $h(x,t)$ the solution of the stochastic heat equation  for $x\in \bb R$,  with zero  initial condition,
\[
\begin{cases} \partial_t h = \frac 12 \partial_x^2 h +\dot W\;, \\  h(x,0) = 0\;, \end{cases}
\]
where $\dot W=\dot W_{x,t}$ is a space-time white noise. 

\begin{proposition}
The process $r(t)$ satisfies,
\begin{itemize}
\item[] $\{r(t)\}_{t\ge 0} \stackrel{\rm law}{=} \{2\,\nu_O^{(\frac 14)}(t)\}_{t\ge 0}\;,$
\item [] $\{r(t)\}_{t\ge 0} \stackrel{\rm law}{=} \{(2\pi)^{\frac 14} h(0,t)\}_ {t\ge 0}\;.$ 
\end{itemize}
\end{proposition}
 
\begin{proof} 
Both statements follow  by computing covariances, since all the processes involved are Gaussian. Indeed, the process $h(x,t)$ is given by 
\[
h(x,t)=\int^t_0 \frac{\rme ^{-\frac{(x-y)^2}{2(t-s)}}}{\sqrt{2\pi (t-s)\big)}}\,\rmd  W_{y,s}\;,
\]
and the covariance is easily computed,
\begin{equation*}
\begin{split}
\bb E \big(h(0,t)\,h(0,t')\big)&=\int _0^{t\land t'}\!\rmd s\int\! \rmd y \, 
\frac{\rme ^{-\frac{y^2}{2(t-s)}}}{\sqrt{2\pi (t-s)\big)}}\, \frac{\rme ^{-\frac{y^2}{2(t'-s)}}}{\sqrt{2\pi (t'-s)\big)}}
\\ &=\frac{1}{\sqrt{2\pi}}\big(\sqrt{t+t'} -\sqrt{t+t'-2(t\land t')}\big)\;.
\end{split}
\end{equation*}
\end{proof}
 
We also obtain a representation as an integral with respect to Brownian motion. 

\begin{proposition}
The process $r$ can be represented as the following integral with respect to Brownian motion $b$, 
\begin{equation}
\label{rep}
r(t)= c \int_0^t\!\frac{u^{\frac14}}{(t^2-u^2)^{\frac14}}\, \rmd b(u)\;,
\end{equation}
where $c= \big(\frac 12 B(\frac 34, \frac 34)\big)^{-\frac12}$, with $B(\cdot,\cdot)$ the usual Euler beta function. 
\end{proposition}

\begin{proof}
Suppose that $t\ge t'\ge 0$. Then, Formula 9.121(4), p.1040 in \cite{GR} reads, 
\begin{equation}
\label{F}
{(t+t')}^{\rho}-{(t-t')}^{\rho}=2\rho\,t'\,t^{\rho-1} F\Big(-\frac{\rho-1}{2},-\frac{\rho-2}{2};\frac32,\frac{{t'}^2}{t^2}\Big)\;,
\end{equation}
where $F$ is the hypergeometric Gauss function. Taking $\rho=\frac12$, and using for $F$ the Integral Formula 9.11.1, p.1040 in \cite{GR} we obtain,
\begin{equation}
F\Big(\frac14,\frac34;\frac32,\frac{{t'}^2}{t^2}\Big) =  \frac{1}{B(\frac34,\frac34)}\int_0^1\!\rmd s\, s^{-\frac14}(1-s)^{-\frac14}\Big(1-\frac{s\,{t'}^2}{t^2}\Big)^{-\frac14}\;.
\end{equation}
The change of variables $u=t's^{\frac 12}$ in this last integral, together with \eqref{F} yields,
\begin{equation}
\sqrt{t+t'}-\sqrt{t-t'}=\frac{2}{B(\frac34,\frac34)}\int^{t'}_0\!\rmd u\, \frac{u^{\frac12}}{(t^2-u^2)^\frac14 ({t'}^2-u^2)^\frac14}\;,
\end{equation}
which proves \eqref{rep}, since the integral on the  right hand side above  is just the covariance of the process given by the stochastic integral  in  \eqref{rep}.  
\end{proof}

\appendix

\section{}
\label{sec:A}

\subsection{Proof of Proposition \ref{hoy}}
\label{A:1}

From \eqref{cov} and \eqref{dyg} 
\begin{equation}
\label{ey2}
\begin{split}
\bb E\, Y(x,t)^2  & = \int_0^t\!\rmd s \int\!\rmd y\, a_{\epsilon}(y)^2 \frac{1}{t-s}\phi'\Big(\frac{x-y}{(t-s)^{\frac 14}}\Big)^2 \\  & = \int_0^t\!\rmd s \frac{1}{s^{\frac 34}} \int\!\rmd z\, a_{\epsilon}^2\big(x-s^{\frac 14}z\big) \phi'(z)^2 \le t^{\frac 14} \|\phi'\|^2_2\;, 
\end{split}
\end{equation}
and then \eqref{vary} follows from \eqref{norma}. 

Before proving the remaining estimates, we remark that \eqref{px} implies that
\begin{equation}
\label{2d0}
|\phi'(x+H)-\phi'(x)| \le C \frac{H}{1+H}\;.
\end{equation}
We then have, 
\begin{equation}
\label{Hx}
\begin{split}
& \bb E \big(Y(x+h,t)-Y(x,t)\big)^2 \\ & \qquad = \int_0^t\!\rmd s\,  \frac{1}{s^{\frac 34}} \int\!\rmd z\, a_{\epsilon}^2\big(x-s^{\frac 14}z\big)  \left[\phi'\big(z+s^{-\frac 14} h\big)-\phi'(z)\right]^2 \\ & \qquad \le  \int_0^t\!\rmd s\,  \frac{1}{s^{\frac 34}} \int\!\rmd z\, \left[\phi'\big(z+s^{-\frac 14} h\big)-\phi'(z)\right]^2 \\ & \qquad = 2 \int_0^t\!\rmd s\,  \frac{1}{s^{\frac 34}} \int\!\rmd z\, \phi'(z) \left[\phi'(z) - \phi'\big(z+s^{-\frac 14} h\big)\right]  \\ & \qquad \le C \int_0^t\!\rmd s\, \frac{1}{s^{\frac 34}}\frac{h}{s^{\frac 14}+h} = 4\,C h\log\big(1+h^{-1}t^{\frac 14}\big)\;,
\end{split} 
\end{equation}
where the last inequality follows from \eqref{norma} and \eqref{2d0} with $H=s^{-\frac 14} h$. This proves the estimate \eqref{hx}. 
 
We next have, after changing variables  in the time integrals,
\[
\begin{split}
& \bb E\big(Y(x,t+h)-Y(x,t)\big)^2 = \int_0^h\!\rmd s \int\!\rmd y\, \frac{a_{\epsilon}(y)^2}{s}\phi'\Big(\frac{x-y}{s^{\frac 14}}\Big)^2 \\ & \qquad + \int_0^t\!\rmd s \int\!\rmd y\,  a^2_{\epsilon}(y) \Big[\frac{1}{(s+h)^{\frac 12}}\phi'\Big(\frac{x-y}{(s+h)^{\frac 14}}\Big) - \frac{1}{s^{\frac 12}}\phi'\Big(\frac{x-y}{s^{\frac 14}}\Big)\Big]^2\;.
\end{split}
\]
Let us denote by $I_1$ and $I_2$ the integrals on the right hand side. The first integral can be bounded as in \eqref{ey2} to obtain
\begin{equation}
\label{ht1}
I_1 \le C\,h^{\frac 14}\;.
\end{equation}
For the second one, we observe that $I_2\le I_{21}+I_{22}$ with
\[
\begin{split}
 I_{21} & = 2 \int_0^t\!\rmd s \int\!\rmd y\, \Big(\frac{1}{s ^{\frac 12}}-\frac{1}{(s+h)^{\frac 12}}\Big)^2 \phi'\Big(\frac{x-y}{s^{\frac 14}}\Big)^2\;, \\ I_{22} & = 2 \int_0^t\!\rmd s \int\!\rmd y\, \frac{1}{s+h}\Big[\phi' \Big(\frac{x-y}{(s+h)^{\frac 14}}\Big) - \phi'\Big(\frac{x-y}{s^{\frac 14}}\Big)\Big]^2\;.
\end{split}
\]
Now, after  the change of variables  $z=s^{-\frac 14}(x-y)$ and then $s= h\tau$ we get,
\begin{equation}
\label{ht2}
I_{21} = 2 \|\phi'\|_2^2\;  h^{\frac 14} \int_0^{h^{-1}t}\!\rmd\tau  \frac{1}{\tau^{\frac 34}}\Big(1-\frac{\tau^{\frac 12}}{(\tau+1)^{\frac 12}}\Big)^2 \le C h^{\frac 14}\;,
\end{equation}
where we used that the last integral in the right-hand side is bounded above by $\int_0^\infty\!\rmd\tau\, (1+\tau)^{-1} \tau^{-\frac 34}<+\infty$. Next, denoting $A=s^{-\frac 14}(s+h)^{\frac 14}$ and using \eqref{pt},
\begin{equation}
\label{ht3}
\begin{split}
I_{22} & = 2 \int_0^t\!\rmd s \int\!\rmd z\, \frac{1}{(s+h)^{\frac 34}}\left[\phi' (z) - \phi'\left(A z\right)\right]^2 \le C \int_0^t\!\rmd s\, \frac{(A-1)^2 }{(s+h)^{\frac 34}}\\ & = Ch^{\frac 14} \int_0^{h^{-1}t}\!\rmd \tau\, \frac{\big[(1+\tau)^{\frac 14}-\tau^{\frac 14}\big]^2}{\tau^{\frac 12} (1+\tau)^{\frac 34}} \le C h^{\frac 14}\;,
\end{split}
\end{equation}
having used that the last integral is bounded by $\int_0^\infty\!\rmd\tau\, (1+\tau)^{-\frac 34} \tau^{-\frac 12}<+\infty$. The estimate \eqref{ht} follows from \eqref{ht1}, \eqref{ht2}, and \eqref{ht3}.
\qed

\subsection{Proof of Proposition \ref{pcov}}
\label{A:2}

From \eqref{defK} and recalling $\|a_\epsilon\|_\infty\le 1$, we have that 
\begin{equation}
\label{eh2}
\begin{split}
\bb E H(x,t)^2 & = \int_0^t\!\rmd s \int\!\rmd y\, \big[\partial_x K (x,y,s) {a_\epsilon(y)}\big]^2 \\ & \le 2 \int^{t\land t_0}_0\!\rmd s \int\!\rmd y\, \big[\big(\partial_x K_\infty(x,y,s)\big)^2 + \big(\partial_x\tilde K(x,y,s)\big)^2\big] \\ & \quad + 2 \int_{t\land t_0}^t\!\rmd s \int\!\rmd y\, \big[\big(\partial_x K^{\ast}(x,y,s)\big)^2+\big(\partial_x k(x,y,s)\big)^2\big]\;.
\end{split}
\end{equation}
Now, from \eqref{Ki} we get that 
\begin{equation}
\label{eki}
\begin{split}
\int^{t\land t_0}_0\!\rmd s \int\!\rmd y\, \big(\partial_x K_\infty(x,y,s)\big)^2 
\le C\int^{t\land t_0}_0\!\rmd s \frac 1s \int\!\rmd y\, \rme^{-(\frac 2 s)^{\frac 14}\,|x-y|}\\ \le C\int^{t\land t_0}_0\!\rmd s\, \frac {1}{s^{\frac 34}} \int\!\rmd z\,  \rme^{-|z|} = C (t\land t_0)^{\frac 14}\;.
\end{split}
\end{equation}
Analogously, from \eqref{Kt} and \eqref{ddt}, 
\begin{align}
\label{tildek}
&\int^{t\land t_0}_0\!\rmd s \int\!\rmd y\, \big(\partial_x\tilde K(x,y,s)\big)^2 
\le C (t\land t_0)^{\frac 34}\;, \\ \label{sk} & \int_{t\land t_0}^t\!\rmd s \int\!\rmd y\,  \big(\partial_x k(x,y,s)\big)^2 \le C\big((t_0\land t)^{-\frac 12} - t^{-\frac 12} \big)\;.
\end{align}
We are left with the estimation of the integral of $(\partial_x K^{\ast})^2$, that gives the leading term (as $t\to \infty$). With the aid of formula \eqref{kstar}, after recalling that $\bam '$, $\varphi$, and $z\mapsto z\rme^{-z^2}$ are bounded  and integrable functions, we obtain,  
\begin{equation}
\label{dxKs}
\begin{split}
& \int_{t\land t_0}^t\!\rmd s \int\!\rmd y\, \big(\partial_x K^{\ast}(x,y,s)\big)^2 \\ &  \qquad \qquad \le C\Big(\bam'(x)^2 \Big( \log\frac{t}{t\land t_0} + t^{\frac 12}-(t_0\land t)^{\frac 12} \Big) +\frac{1}{t\land t_0}-\frac{1}{t}\Big)\;,
\end{split}
\end{equation}
and \eqref{varH} follows from \eqref{eh2}, \eqref{eki}, \eqref{tildek}, \eqref{sk}, and the above estimate. 

To prove the rest of the proposition, let us use formula \eqref{HK} for $H$, to write
\begin{equation}
\label{dxH}
\begin{split}
& \bb E \big( H(x+h,t)- H(x,t)\big)^2 \le 2\, \bb E \big( Y(x+h,t)- Y(x,t)\big)^2 \\ & \qquad\qquad\qquad+ 2\, \bb E \big( \mc G (V''(\bam) H)(x+h,t) -  \mc G (V''(\bam) H)(x,t)\big)^2\;.
\end{split}
\end{equation}
Recalling \eqref{defcalg} (with the substitution $s\to t-s$) and using $V''(\bam) \le 3\bam ^2 +1 \le 4$, we estimate the last expectation above by Cauchy Schwartz inequality, 
\begin{equation}
\label{CS}
\bb E\big(\mc G (V''(\bam) H)(x+h,t)-  \mc G (V''(\bam)H)(x,t)\big)^2 \le J(x,t)^2\;,
\end{equation}
where
\[
J(x,t) = 4 \int_0^t\!\rmd s \int\!\rmd y\, \big| \partial_y^2G(x+h,y,s)-\partial_y^2G(x,y,s)\big| \big(\bb E  H(y,t-s)^2\big)^{\frac 12}\;.
\]
From \eqref{varH} we have $\big(\bb E  H(y,t-s)^2\big)^{\frac 12} \le C\,\big( 1 + \bam'(y) (t-s)^{\frac 14} \id_{\{t-s>1\}}\big)$. Then, by formula \eqref{dyg} for $\partial_y^2G$, after changing variable 
$y=x-s^{\frac 14} z$ we get, 
\[
\begin{split}
& J(x,t) \le C \int_0^t\!\rmd s\, \frac{1}{s^{\frac 12}} \int\!\rmd z\, \big|\phi''\big(z+s^{-\frac 14}h\big)-\phi''\big(z\big)\big| \\ & \quad\;  + C t^{\frac 14}\int_{0}^{(t-1)_+}\!\rmd s\, \frac{1}{s^{\frac 12}} \int\!\rmd z\, \big|\phi''\big(z+s^{-\frac 14}h\big)-\phi''\big(z\big)\big| \bam'(x-s^{\frac 14} z)\,,
\end{split}
\]
which implies, as $\bam'$ is bounded,
\begin{equation}
\label{T123}
\begin{split}
& J(x,t) \le C \int_0^t\!\rmd s\, \frac{1}{s^{\frac 12}} \int\!\rmd z\, \big|\phi''\big(z+s^{-\frac 14}h\big)-\phi''\big(z\big)\big| \\ & \quad\;  + C t^{\frac 14}\int_{0}^{1\land (t-1)_+}\!\rmd s\, \frac{1}{s^{\frac 12}} \int\!\rmd z\, \big|\phi''\big(z+s^{-\frac 14}h\big)-\phi''\big(z\big)\big| \\ & \quad\;  + C t^{\frac 14}\int_{1\land (t-1)_+}^{(t-1)_+}\!\rmd s\, \frac{1}{s^{\frac 12}} \int\!\rmd z\, \big|\phi''\big(z+s^{-\frac 14}h\big)-\phi''\big(z\big)\big| \bam'(x-s^{\frac 14} z)\,,
\end{split}
\end{equation}
Now, by \eqref{px} with $H = s^{-\frac 14}h$ and $k=2$, 
\[
\int\!\rmd z\, \big|\phi''\big(z+s^{-\frac 14}h\big)-\phi''\big(z\big)\big| \le C s^{-\frac 14}h\;, 
\]
while, by \eqref{2d1} and using that  $\bam'$ vanishes exponentially fast at $\pm\infty$,
\[
\int\!\rmd z\, \big|\phi''\big(z+s^{-\frac 14}h\big)-\phi''\big(z\big)\big| \bam'(x-s^{\frac 14} z) \le C s^{-\frac 34} \big(h^2 + (1+|x|) h\big)\;.
\]
Using the above estimates in the right hand side of \eqref{T123} we obtain,
\begin{equation}
\label{I<}
J(x,t) \le C  t^{\frac 14} \big(h^2 + (1+|x|) h\big)\;.
\end{equation}
Estimate \eqref{hxH} follows now  from  \eqref{dxH}, \eqref{hx}, \eqref{CS}, and \eqref{I<}.

\smallskip
A similar reasoning proves \eqref{htH}. Indeed, by formula \eqref{HK} for $H$ we have in this case,
\begin{equation}
\label{dtH}
\begin{split}
& \bb E \big( H(x,t+h)- H(x,t)\big)^2 \le 2\, \bb E \big( Y(x,t+h)- Y(x,t)\big)^2 \\ & \qquad\qquad\qquad+ 2\, \bb E \big( \mc G (V''(\bam) H)(x,t+h) -  \mc G (V''(\bam) H)(x,t)\big)^2\;,
\end{split}
\end{equation}
with now (as before, recalling \eqref{defcalg}, that $V''(\bam) \le 3\bam ^2 +1 \le 4$, changing $s\to t-s$, and using Cauchy Schwartz inequality),
\begin{equation}
\label{CSt}
\begin{split}
&\bb E\big(\mc G (V''(\bam) H)(x,t+h)-  \mc G (V''(\bam)H)(x,t)\big)^2 \le 2 J_1(x,t)^2 + 2 J_2(x,t)^2\;,
\end{split}
\end{equation}
where
\[
J_1(x,t) = 4 \int_0^h\!\rmd s \int\!\rmd y\, \big|\partial_y^2 G(x,y,s)\big| \big(\bb E  H(y,t + h - s)^2\big)^{\frac 12}
\]
and 
\[
J_2(x,t) = 4 \int_0^t\!\rmd s \int\!\rmd y\, \big| \partial_y^2G(x,y,s+h)-\partial_y^2G(x,y,s)\big| \big(\bb E  H(y,t-s)^2\big)^{\frac 12}\;.
\]
Proceeding as in \eqref{CS}, from \eqref{dyg}, \eqref{norma} and \eqref{varH} we obtain  
\begin{equation}
\label{CSt1}
\begin{split}
J_1(x,t) & \le 4\! \sup_{z\in \bb R,\, \tau\le t+h}\big(\bb E\,H(y,\tau)^2\big)^{\frac 12} \int_0^h\!\rmd s \frac{1}{s^{\frac 12}} \int\!\rmd z\,  \big|\phi''(z)\big|  \\ & \le  C  h^{\frac 12} \big(1+h^{\frac 14}+t^{\frac 14}\big)\;. 
\end{split}
\end{equation}
Recalling \eqref{varH} implies $\big(\bb E  H(y,t-s)^2\big)^{\frac 12} \le 1 + \bam'(y) (t-s)^{\frac 14} \id_{\{t-s>1\}}$,  we now have,
\[
\begin{split}\;\;
J_2&(x,t) \le 4 \int_0^t\!\rmd s\, \frac{1}{(s+h)^{\frac 34}} \int\!\rmd y\, \Big|\phi''\Big(\frac{x-y}{(s+h)^{\frac 14}}\Big)-\phi''\Big(\frac{x-y}{s^{\frac 14}}\Big)\Big| \\ &   + 4 \int_0^t\!\rmd s\, \Big(\frac{1}{s^{\frac34}} -\frac{1}{(s+h)^{\frac34}} \Big)\int\!\rmd y\,  \Big|\phi''\Big(\frac{x-y}{s^{\frac 14}}\Big)\Big|  \\ &  + 4\,t^{\frac 14}  \int_0^{(t-1)_+}\!\rmd s\, \frac{1}{(s+h)^{\frac 34}} \int\!\rmd y\, \Big|\phi''\Big(\frac{x-y}{(s+h)^{\frac 14}}\Big)-\phi''\Big(\frac{x-y}{s^{\frac 14}}\Big)\Big| \bam'(y) \\ &   + 4\,t^{\frac 14} \int_0^{(t-1)_+}\!\rmd s\, \Big(\frac{1}{s^{\frac34}} -\frac{1}{(s+h)^{\frac34}} \Big)\int\!\rmd y\,  \Big|\phi''\Big(\frac{x-y}{s^{\frac 14}}\Big)\Big| \bam'(y)\;.
\end{split}
\]
We change variables $y=x-(s+h)^{\frac 14}z$ in the first and third integral and 
$y=x-s^{\frac 14}z$ in the second integral, and shorthand $A=\big(\frac{s+h}{s}\big)^{\frac 14}$ to obtain,
\[
\begin{split}
J_2(x,t) & \le 4\int_0^t\!\rmd s\, \frac{1}{(s+h)^{\frac 12}} \int\!\rmd z\, \big|\phi''(z) - \phi''(Az)\big| \\ &   + 4\int_0^t\!\rmd s\, \frac{1}{s^{\frac12}}\Big(1 -\frac{s^{\frac34}}{(s+h)^{\frac34}} \Big) \int\!\rmd z\, \big|\phi''(z) \big|  \\ &  +
4\,t^{\frac 14}  \int_0^{(t-1)_+}\!\rmd s\, \frac{1}{(s+h)^{\frac 12}} \int\!\rmd z\, \big|\phi''(z) - \phi''(Az)\big| \bam'(x-(s+h)^{\frac 14}z)\\ &   +4\,t^{\frac 14} \int_0^{(t-1)_+}\!\rmd s\, \Big(\frac{1}{s^{\frac34}} -\frac{1}{(s+h)^{\frac34}} \Big)\int\!\rmd y\,  \Big|\phi''\Big(\frac{x-y}{s^{\frac 14}}\Big)\Big| \bam'(y)\;.
\end{split}
\]
We denote by $J_{21}, J_{22}, J_{23}, J_{24}$ the four integrals on the right hand side. By \eqref{pt} and noticing \begin{equation}
\label{A-1}
\begin{split}
A-1 & = \frac{(s+h)^{\frac 14}-s^{\frac 14}}{s^{\frac 14}}  = \frac{(s+h)^{\frac 12}-s^{\frac 12}}{s^{\frac 14} \big((s+h)^{\frac 14}+s^{\frac 14}\big)}  \\ & =   \frac{h}{s^{\frac 14}\big((s+h)^{\frac 14}+s^{\frac 14}\big)\big((s+h)^{\frac 12}+s^{\frac 12}\big)}  \le \Big(\frac h s\Big)^{\frac 34}\;,
\end{split}
\end{equation}
\[
J_{21} \le C\int_0^t\!\rmd s\, \frac{1}{(s+h)^{\frac 12}}\Big(\frac{h}{s}\Big)^{\frac 14} = C h^{\frac 12} \int_0^{h^{-1}t} \!\rmd\tau\, \frac{1}{\tau^{\frac 34} (1+\tau)^{\frac 12}} \le Ch^{\frac 12}\;,
\]
while, as $\phi''$ is integrable,
\[
J_{22} \le C\int_0^t\!\rmd s\, \frac{1}{s^{\frac12}}\Big(1 -\frac{s^{\frac34}}{(s+h)^{\frac34}} \Big)  = C h^{\frac 12} \int_0^{h^{-1}t} \!\rmd\tau\, \frac{(1+\tau)^{\frac 34}-\tau^{\frac 34} }{\tau^{\frac 12}(1+\tau)^{\frac 34}} \le Ch^{\frac 12}\;.
\]
To estimate $J_{23}$ we observe that \eqref{pt} implies $\big|\phi''(z) - \phi''(Az)\big| \le C\frac{A-1}A$, and $\rmd z$-integration of $\bam'(x+(s+h)^{\frac 14}z)$ gives an extra factor $(s+h)^{-\frac 14}$. Therefore,
\[
 J_{23} \le C t^{\frac 14} \int_0^t\!  \rmd s\, \frac{1}{(s+h)^{\frac 34}}
  \Big(\frac{h}{s}\Big)^{\frac 34}
 =  C(ht)^{\frac 14} \int_0^{h^{-1}t}\!\rmd\tau\, \frac{1}{\tau^{\frac 34}\,(1+\tau)^{\frac 34}}
  \le C(ht)^{\frac 14}\;. 
\]
Analogously, as $\phi''$ is bounded and $\bam'$ is integrable we finally have,
\[
J_{24} \le  C(ht)^{\frac 14}  \int_0^{h^{-1}t} \!\rmd\tau\, \frac{(1+\tau)^{\frac 34}-\tau^{\frac 34} }{\tau^{\frac 34}(1+\tau)^{\frac 34}} \le C(ht)^{\frac 14} \;.
\]
We conclude that
\begin{equation}
\label{J2}
J_2(x,t) \le C \big( h^{\frac 12} + (ht)^{\frac 14}\big)\;.
\end{equation}
The estimate \eqref{htH} follows from \eqref{dtH}, \eqref{ht}, \eqref{CSt}, \eqref{CSt1}, and \eqref{J2}.
\qed

\section*{Acknowledgements} We are indebted to Errico Presutti for stimulating and helpful discussions on the subject of this article. S.B. gratefully  aknowledges the kind hospitality and support  of the Department of Mathematics of the University of Rome La Sapienza,  and of the SPDEs programme held at the Isaac Newton Institute, where part of this research was conducted.

\end{document}